\documentclass[journal]{IEEEtran}
\usepackage[colorlinks,urlcolor=blue,linkcolor=blue,citecolor=blue]{hyperref}
\usepackage{amsmath,amssymb,amsfonts,bbm,mathrsfs,mathtools,amsthm}
\usepackage{color,array}
\usepackage{tikz}
\usepackage{graphicx}
\usepackage{booktabs}

\newtheorem{theorem}{Theorem}
\newtheorem{lemma}{Lemma}

\newtheorem{definition}{Definition}
\newtheorem{assumption}{Assumption}
\newtheorem{remark}{Remark}

\newtheorem{proposition}{Proposition}
\newcommand{\norm}[1]{{\left \lVert #1 \right \rVert}}
\newcommand{\abs}[1]{{\left \lvert #1 \right \rvert}}
\newcommand{\rank}[1]{{\mathrm{rank} (#1) }}
\definecolor{bl}{rgb}{0,0.45,0.60}
\definecolor{red}{rgb}{0.77,0.01,0.2}
\definecolor{pur}{rgb}{0.5,0.01,0.99}

\begin{document}
\title{Local Observability and Moving Horizon Estimation-based Training of Feedforward \\Neural Networks}
\author{Yi Yang, Victor G. Lopez, and  Matthias A. M\"uller
\thanks{Y. Yang,  V. G. Lopez, and M. A. M\"{u}ller are with the Leibniz University Hannover, Institute of Automatic Control, 30167 Hannover, Germany\\ (e-mail: \{yang, lopez, mueller\}@irt.uni-hannover.de). }
\thanks{This work was supported by the Deutsche Forschungsgemeinschaft (DFG, German Research Foundation) - 535860958.}
}

\maketitle

\begin{abstract}
In this paper, we propose a moving horizon estimation (MHE)-based training method for feedforward neural networks~(FNNs) with rectified linear unit~(ReLU) activation functions to determine their ideal weights from a control-theoretic perspective. This allows for a rigorous theoretical analysis of the trained network. First, we reformulate the FNN as a dynamical system with the weights as states. Then, we investigate the local observability of such a system. For two-layer FNNs with fixed output weights, we derive a sufficient condition under which the observability rank condition holds, ensuring a locally observable state. We also show that multi-layer FNNs in general fail to satisfy the observability rank condition. Based on this analysis,  we develop a persistently exciting (PE) input design method, which renders a state distinguishable from its neighbors.  The resulting local observability provides convergence guarantees for the proposed MHE-based training, where only the projection of the state onto the observable subspace is updated using a fixed-length window of input-output data. The effectiveness of the approach is illustrated via numerical examples.
\end{abstract}


\section{Introduction}
Feedforward neural networks (FNNs) are a class of fundamental neural network architectures, where information flows in only one direction. By incorporating nonlinear activation functions, FNNs serve as powerful universal function approximators \cite{bonassi2022recurrent,hunt1992neural}. They are widely employed in various network architectures, including convolutional neural networks \cite{krizhevsky2017imagenet} and transformers \cite{vaswani2017attention}. In control engineering, FNNs can be used to approximate control laws, e.g., to generate control signals that emulate model predictive controllers \cite{11159581,9993046,li2022using}.  

In practice, the backpropagation algorithm with gradient descent (GD) technique and its variants  have become the most commonly used methods for neural network training \cite{schmidhuber2015deep}. Despite practical success, theoretical guarantees for convergence to the optimal weights are rarely studied \cite{li2017convergence}. Beyond backpropagation, state estimation-based methods have been explored to train NNs, where the weights of an NN are treated as the states of a nonlinear system.  An extended Kalman filter has been adapted for NN training through updating weights incrementally in \cite{singhal1988training}, improving the convergence speed in practical examples. This approach was later extended to general recurrent neural network learning in \cite{bemporad2022recurrent}, where it demonstrated performance competitive with stochastic gradient descent in a nonlinear identification problem. In addition to the extended Kalman filter, moving horizon estimation (MHE), which is a powerful optimization-based state estimation method with strong theoretical guarantees on convergence \cite{schiller2023lyapunov}, has also been explored for NN training. In \cite{bonassi2022towards}, an online weight adaptation approach based on MHE was developed  for lifelong training of recurrent neural networks. However, a key assumption requiring the system states (i.e., the weights of neural networks) to be observable or at least detectable is not always satisfied, leaving open problems in developing a more thorough theoretical analysis. The optimal weights of an FNN are in general non-unique \cite{kleinberg2018alternative}, implying that global observability of its corresponding dynamical system cannot be guaranteed. In \cite{albertini1993uniqueness}, equivalence of two-layer FNNs up to a finite number of sign flips and node permutations was established under sufficient conditions requiring the activation functions to be odd and linearly independent over all inputs. In \cite{bona2023parameter}, the authors extended the analysis of parameter equivalence to deep ReLU networks and proposed a sufficient condition on the FNN structure and weights, under which any two FNNs are considered equivalent. The lack of global observability motivates us to instead  investigate a weaker form of observability, namely local observability, for NN weights.  

Local observability is a central concept in the analysis of nonlinear systems, describing the ability to distinguish a state from its neighbors using input-output measurements. The seminal work of Hermann and Krener \cite{hermann1977nonlinear} established the theoretical foundations of local observability and local controllability for continuous-time nonlinear systems, and introduced a sufficient condition for local observability, called observability rank condition.  These results were subsequently extended to discrete-time systems in \cite{nijmeijer1982observability} and \cite{albertini1996remarks}, and the effects of sampling were studied in \cite{sontag1984concept}.  Moreover,  a numerical approach based on an empirical observability Gramian was proposed in \cite{5400067}, enabling a quantitative assessment of local observability for nonlinear systems. This framework was later extended to systems with inputs in \cite{7403218}.  Unlike linear systems, observability of nonlinear systems may depend on the applied control inputs, and certain inputs  are required to distinguish between states. Consequently, input selection is also crucial in observability analysis. However, the systematic design of input sequences for  locally observable states remains a challenging problem \cite{7403218}. 

In this paper, we explore FNN training via MHE, wherein the optimal weights of an FNN are treated as the state of its reformulated dynamical system to be estimated. We first investigate local observability of this system, which is crucial for guaranteeing the convergence of the MHE-based FNN training algorithm.  In \cite{vanelli2025local}, local identifiability of bias-free feedforward networks was investigated for all weights except for those lying on a  set of measure zero under analytic activation functions. In contrast, we consider FNNs with the widely used ReLU activation function, which does not satisfy the restrictions imposed in aforementioned works \cite{albertini1993uniqueness,vanelli2025local}.  In particular, for a specific class of two-layer FNNs with fixed output weights, we derive a sufficient condition that guarantees satisfaction of the observability rank condition at a given state, and thus ensures its local observability. We further examine general multi-layer FNNs with ReLU activation functions, demonstrating that they typically fail to satisfy the observability rank condition.

Based on these results, we develop a systematic approach to design PE input sequences for a locally observable state. Since our objective is to estimate the ideal FNN weights via MHE, the ideal state (i.e., the ideal weights of FNNs) is unknown a priori, and its local observability cannot be verified. To address this, we construct a locally observable neighborhood of a given locally observable state by requiring that  all states inside share the same Jacobian matrix of the observability mapping. This construction ensures that every state within this neighborhood is locally observable and distinguishable from each other under the same PE input sequence. 

Leveraging the  local observability results above, we propose an MHE-based training method for FNNs, where the weights are updated using a fixed-length window of input-output data. In particular, under the assumption that the training dataset is PE, we provide a systematic analysis and prove convergence of the state estimates when the MHE-based training is performed iteratively over these windows.

In summary, the contributions of this work are as follows: (i) reformulating FNNs as dynamical systems and deriving a sufficient condition that guarantees local observability of the corresponding states; (ii) developing a systematic approach for PE input design for locally observable states; (iii) establishing an MHE-based FNN training method with theoretical analysis of its convergence. Preliminary results of this work were presented in the conference paper \cite{yang2025local}, which focuses on local observability of FNNs with equal numbers of input and hidden nodes. In this work, we significantly extend the analysis to general FNN architectures and develop an MHE-based FNN training algorithm with convergence guarantees. 

The remainder of this paper is organized as follows. Section~\ref{section-preliminary} introduces  preliminary results on local observability, FNNs, and elementary vector theory.  Section~\ref{section-observability} investigates local observability of general FNNs, and develops a systematic PE input design together with a locally observable neighborhood. Section~\ref{section-MHE} develops an MHE-based training method for FNNs and establishes convergence guarantees. Finally, Section~\ref{section-simulation} validates the effectiveness of the proposed method numerically.

\section{Preliminaries}\label{section-preliminary}
We denote the set of integers in the interval $[a,b]\subset\mathbb{Z}$ by $\mathbb{Z}_{[a,b]}$, and the set of non-negative integers  by $\mathbb{Z}_{\geq 0}$. Denote by $x_{[a,b]}$ the sequence $\{x_a,x_{a+1},\ldots,x_b\}$.  Denote by $\mathbf{1}_n$ the all-ones vector in $\mathbb{R}^n$, and by $\mathbf{1}_{m\times n}$ the all-ones matrix in $\mathbb{R}^{m\times n}$. Let $\otimes$ denote the Kronecker product between two matrices of arbitrary dimensions, and $\circ$ the element-wise product between two matrices of the same dimensions. For a vector $x\in \mathbb{R}^n$, denote by $\mathrm{diag}(x)$ the diagonal matrix with the entries of $x$ on its main diagonal. The identity matrix of dimension $n$ is denoted by $\mathrm{I}_n$. For a matrix $A\in \mathbb{R}^{m\times n}$, denote by $\mathcal{R}(A)$ the subspace spanned by the rows of $A$. Let $A^\dagger$ denote the Moore-Penrose inverse of $A$. For a continuously differentiable function $f(x): \mathbb{R}^{n}\rightarrow\mathbb{R}^{p}$,  the Jacobian matrix at a point $x$ is denoted by $\mathrm{D}f(x)=\frac{\partial f}{\partial x}(x)$. 
\subsection{Local Observability in Nonlinear Systems}
In the following, we introduce and discuss the notion of local observability for a nonlinear system of the form: 
\begin{subequations}\label{2-eq-nonlinearplant}
	\begin{align}
		x_{t+1}&=f(x_{t},u_{t}),\\
		y_{t}&=h(x_{t},u_{t}),\label{2-eq-nonlinearplant-2}
	\end{align}
\end{subequations}where $t\in\mathbb{Z}_{\geq 0}$, $x_t\in \mathbb{R}^n$, $u_t\in \mathbb{R}^m$, and $y_t\in \mathbb{R}^p$.

\begin{definition}
	Define the $k$-observability mapping $\mathscr{H}_k$: $ \mathbb{R}^n\times (\mathbb{R}^m)^k\rightarrow(\mathbb{R}^{p})^k$ at state $x$ by
	\begin{align}
		\mathscr{H}_k(x,u_{[1,k]})=\begin{bmatrix}
			h(x,u_1)\\
			h(f(x,u_1),u_2)\\
			\vdots\\
			h(f(\cdots f(f(x,u_1),u_2),\cdots,u_k))
		\end{bmatrix},
	\end{align} for any input sequence $u_{[1,k]}$ with some $k\geq 1$. 
\end{definition}

\begin{definition}\label{def-local observable}
	System~(\ref{2-eq-nonlinearplant}) is said to be $k$-observable at $x$ if there exists an input sequence $u_{[1,k]}\in (\mathbb{R}^m)^k$ for some $k\geq 1$ and a neighborhood $\mathcal{M}$ of $x$   such that, for any $x^{\prime}\in \mathcal{M}$, $\mathscr{H}_i(x,u_{[1,i]})=\mathscr{H}_i(x^{\prime},u_{[1,i]})$ for $i=1,\ldots,k$, implies $x^{\prime}= x$.
\end{definition}

The above definition of $k$-observability implies that $x$ is distinguishable from all other states in  $\mathcal{M}$. A definition termed strong local observability was given in \cite{nijmeijer1982observability}, where $k=n$ was considered. Here, we consider a more general positive constant~$k$. 

For ease of notation, we write $\mathscr{H}_k(x)$ instead of $\mathscr{H}_k(x,u_{[1,k]})$ to denote the $k$-observability mapping in the following sections. Before we introduce a sufficient condition for local observability, we present the notion of \textit{observability rank condition} \cite{hermann1977nonlinear}.  To this end, given some $x$, we assume that $f$ and $h$ are continuously differentiable on a neighborhood $\mathcal{M}$ of $x$.
System~(\ref{2-eq-nonlinearplant}) is said to satisfy the observability rank condition  at state $x$ if 
\begin{align}\label{2-eq-ORC}
	\rank{\mathrm{D}\mathscr{H}_k(x)}=n,
\end{align}which means the Jacobian matrix of the $k$-observability mapping with respect to state $x$ has full column rank.   With the definitions above, we introduce the following lemma.
\begin{lemma}[\!\!{\cite[Th. 5]{albertini1996remarks}}]\label{2-lem-ORC}
	If system~(\ref{2-eq-nonlinearplant}) satisfies the observability rank condition at $x$, then it is $k$-observable at $x$.
\end{lemma} 

Lemma \ref{2-lem-ORC} establishes a sufficient condition for $k$-observability, which relies on the first-order derivative of the observability mapping. Specifically, satisfaction of the rank condition in (\ref{2-eq-ORC}) implies that the mapping is injective locally around $x$. This result was first obtained in \cite{hermann1977nonlinear} using a Lie-algebraic  characterization for nonlinear continuous-time systems. It was later adapted to autonomous discrete-time systems in \cite{nijmeijer1982observability}, and it was shown to be \textit{sufficient and necessary} for local observability under the assumption that $\mathrm{D}\mathscr{H}_k(x)$ is constant dimensional in a neighborhood of $x$  in \cite{albertini1996remarks}. 

\begin{definition}
	An input sequence $u_{[1,k]}$ is said to be persistently exciting if the observability rank condition holds at state $x$ under $u_{[1,k]}$.
\end{definition}

 For convenience, throughout the remainder of this work, we use \textit{locally observable} to refer to $k$-observable, and we write \textit{observability mapping} in place of $k$-observability mapping.
\subsection{Feedforward neural networks (FNNs) with ReLU activation functions}\label{section-pre-FNN}
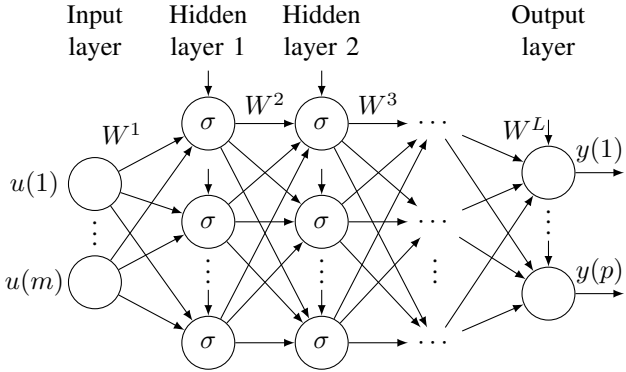
\begin{figure}[htb]
	\centering
	\begin{tikzpicture}[auto, every text node part/.style={align=center},
		neuron/.style={circle, draw, minimum size=0.7cm}, 
		connect/.style={draw, -latex},
		layer/.style={align=center},
		midpoint/.style = {coordinate}
		]
		\node[neuron] (I1) at (0, 0.65) {}; 
		\node at ([xshift=-0.8cm] I1) {$u(1)$};
		\node[neuron] (I2) at (0, -0.65) {}; 
		\node at ([xshift=-0.8cm] I2) {$u(m)$};
		
		\node at (0, 0.1) {\vdots}; 
		
		\node[neuron] (H1) at (1.5, 1.45) {$\sigma$}; 
		\node[neuron] (H2) at (1.5, 0.15) {$\sigma$}; 
		\node at (1.5, -0.4) {\vdots}; 
		\node[neuron] (H3) at (1.5, -1.45) {$\sigma$}; 
		
		\node[neuron] (H21) at (3, 1.45) {$\sigma$}; 
		\node[neuron] (H22) at (3, 0.15) {$\sigma$}; 
		\node at (3, -0.4) {\vdots}; 
		\node[neuron] (H23) at (3, -1.45) {$\sigma$}; 
		
		\node (H31) at (4.5, 1.45) {$\cdots$}; 
		\node (H32) at (4.5, 0.15) {$\cdots$}; 
		\node at (4.5, -0.4) {\vdots}; 
		\node (H33) at (4.5, -1.45) {$\cdots$}; 
		
		\node[neuron] (O1) at (6, 0.8) {};
		\node[neuron] (O2) at (6, -0.8) {};
		\node at (6, 0.25) {\vdots}; 

		\foreach \j in {2,3} {
			\foreach \i in {1,2} {
				\draw[connect] (I\i) -- (H\j);
			}
		}
		\draw[connect] (I1) -- node{$W^1$}(H1);
		\draw[connect] (I2) -- (H1);
		
		\draw[connect] (H1) -- node{$W^2$}(H21);
		\foreach \i in {2,3}{
			\draw [connect](H1) -- (H2\i);}
		
		\foreach \i in {2,3}{
			\foreach \j in {1,2,3}{
				\draw [connect] (H\i) -- (H2\j);
			}
		}
		
		\draw[connect] (H21) -- node{$W^3$}(H31);
		\foreach \i in {2,3}{
			\draw [connect](H21) -- (H3\i);}
		
		\foreach \i in {2,3}{
			\foreach \j in {1,2,3}{
				\draw [connect] (H2\i) -- (H3\j);
			}
		}
		
		\draw[connect] (H31) -- node{$W^L$}(O1);
		\foreach \i in {2,3} {
			\draw[connect] (H3\i) -- (O1);
		}
		\foreach \i in {1,2,3} {
			\draw[connect] (H3\i) -- (O2);
		}
		\node [midpoint, right of=O1, node distance=1cm] (output1) {};
		\draw[connect] (O1) -- node{$y(1)$}(output1);
		\node [midpoint, right of=O2, node distance=1cm] (output2) {};
		\draw[connect] (O2) -- node{$y(p)$}(output2);
		\node[layer, above of=I1, node distance=2cm] (input) {Input \\ layer};
		\node[layer, above of=H1, node distance=1.2cm] (hidden) {Hidden \\ layer 1};
		\node[layer, above of=H21, node distance=1.2cm] (hidden) {Hidden \\ layer 2};
		\node[layer, above of=O1, node distance=1.85cm] (output) {Output \\ layer};
		
		\node [midpoint, above of=H1, node distance=0.7cm] (b1) {};
		\node [midpoint, above of=H2, node distance=0.7cm] (b2) {};
		\node [midpoint, above of=H3, node distance=0.7cm] (b3) {};
		\draw[connect] (b1) -- node{}(H1);
		\draw[connect] (b2) -- node{}(H2);
		\draw[connect] (b3) -- node{}(H3);
		
		\node [midpoint, above of=H21, node distance=0.7cm] (b21) {};
		\node [midpoint, above of=H22, node distance=0.7cm] (b22) {};
		\node [midpoint, above of=H23, node distance=0.7cm] (b23) {};
		\draw[connect] (b21) -- node{}(H21);
		\draw[connect] (b22) -- node{}(H22);
		\draw[connect] (b23) -- node{}(H23);
		
		\node [midpoint, above of=O1, node distance=0.7cm] (b31) {};
		\node [midpoint, above of=O2, node distance=0.7cm] (b32) {};
		\draw[connect] (b31) -- node{}(O1);
		\draw[connect] (b32) -- node{}(O2);
	\end{tikzpicture}
	\caption{Architecture of FNNs}
	\label{fig:FNN}
\end{figure}
Consider $L$-layer fully connected FNNs as shown in Fig.~\ref{fig:FNN}, where $u(i), i\in\mathbb{Z}_{[1,m]}$ refers to the $i$th element of an input vector $u\in\mathbb{R}^{m}$, $y(j), j\in\mathbb{Z}_{[1,p]}$ refers to the $j$th element of an output vector $y\in\mathbb{R}^{p}$, and $\sigma(a)=\max(0,a)$ is the ReLU activation function (when $a$ is a vector, $\sigma$ is applied element-wise). The function $h(W,u)$ between the input $u\in\mathbb{R}^m$ and the output $y\in\mathbb{R}^{p}$ is given by 
\begin{align}\label{2-eq:FNN}
y\!=\!	h(W,u)\!=\!W^L\sigma(\!\cdots(\!\sigma(\!W^{1}u\!+\!b^1)\!+\!\cdots\!)\!+\!b^{L-1})\!+\!b^L,
\end{align}
where $W^{i}\in\mathbb{R}^{n_{i-1}\times n_i}$ and $b^i\in\mathbb{R}^{n_i}$, for $i\in \mathbb{Z}_{[1,L]}$, are weight matrix and bias weights of the $i$th layer, respectively, with $n_i$ denoting the number of nodes in the $i$th layer (the input layer is the 0th layer). For the output layer, the activation function is the identity function.    

\subsection{Orthants and elementary vectors}
In the following, we introduce some preliminary results on orthants and elementary vectors that will be used for the subsequent local observability analysis from a subspace perspective. Define an open orthant in $n$-dimensional space as
\begin{align}\label{def-orthant}
	\mathcal{O}_{s}\coloneqq \left\{x\in\mathbb{R}^{n} \mid x=\sum_{i=1}^{n}\alpha_x(i) s(i) e_i, \alpha_x(i)>0\right\},
\end{align}where $e_i$, $i\in\mathbb{Z}_{[1,n]}$, denote the standard orthonormal basis for $\mathbb{R}^n$, and $s=(s(1),\ldots,s(n))\in\{-1,1\}^{1\times n}$ is the sign vector of $\mathcal{O}_{s}$ representing the sign of each coordinate axis of this orthant. By definition, when $n=1$ and $n=2$, the orthant is called the ray and the quadrant, respectively.

Let the row space of $V\in\mathbb{R}^{m\times n}$ have a nonempty intersection with $r$ distinct open orthants in $\mathbb{R}^{n}$. Denote these orthants by $\mathcal{O}_{s_i}$, for  $i\in \mathbb{Z}_{[1,r]}$, where each $s_i$ is the sign vector of $\mathcal{O}_{s_i}$ as defined in~(\ref{def-orthant}). Define $S_{RV}\coloneqq [s_1^{\top},s_2^{\top}, \ldots,s_r^{\top}]^{\top}\in \{-1,1\}^{r\times n}$, a sign matrix of $\mathcal{R}(V)$. Note that  for any row $s_i$ of $S_{RV}$, there exists a vector $v\in \mathcal{R}(V)$ such that $\mathrm{sign}(v)=s_i$. However, the converse does not hold because the row space of $V$  consists not only of vectors lying in the open orthants that intersect $\mathcal{R}(V)$ but also  on the boundaries between them. 

In the following, we introduce the concept of elementary vector.
\begin{definition}[\!\!\cite{aichmayr2024sagemath}]
	Consider a subspace $\mathcal{R}(V)\subseteq\mathbb{R}^{n}$. For a nonzero vector $v\in \mathcal{R}(V)$, we define its \textit{support}  as: $\mathrm{supp}(v)\coloneqq\{i\in \mathbb{Z}_{[1,n]}\mid v(i)\neq 0\}$, which is the index set of its nonzero entries. A vector $v$ is said to be an \textit{elementary vector} if, for any nonzero vectors $v^{\prime}\in \mathcal{R}(V)$, $\mathrm{supp}(v^{\prime})\subseteq \mathrm{supp}(v)$ implies $\mathrm{supp}(v^{\prime})= \mathrm{supp}(v)$.
\end{definition}

By the above definition, the support of an elementary vector is minimal with respect to inclusion. For example, consider $V=\begin{bmatrix}
	1 &-1& 0\\0&0&1
\end{bmatrix}$, supports of vectors $(1,-1,0)$ and $(0,0,1)$ are $\{1,2\}$ and $\{3\}$. Since the support of any vector in $ \mathcal{R}(V)$ is one of $\{1,2\}$, $\{3\}$, or $\{1,2,3\}$, the supports $\{1,2\}$ and $\{3\}$ are minimal. Consequently, $(1,-1,0)$ and $(0,0,1)$ are elementary vectors of $ \mathcal{R}(V)$.  Moreover, based on results in \cite{rockafellar1969elementary}, the set of elementary vectors of a subspace is finite and  any vector in this subspace can be expressed as a linear combination of them. Procedures for computing elementary vectors of a subspace are provided in \cite{brualdi1995sparse}. A package was developed in \cite{aichmayr2024sagemath} for computing elementary vectors and all $\mathrm{sign}(v)$ for $v\in \mathcal{R}(V)$. 

Next, we develop a proposition regarding the intersection of a subspace and an open orthant, which forms a cone used in the following sections.

\begin{proposition}\label{3-proposition-orthant and subspace reconstruction}
	Consider a matrix $V\in\mathbb{R}^{m\times n}$ with no zero columns, and let $\mathrm{rank}(V)=r\leq m$. Suppose its row space $\mathcal{R}(V)$ intersects an orthant $\mathcal{O}_{s}$. Then, the intersection $\mathcal{R}(V)\cap \mathcal{O}_{s}$ contains $r$ linearly independent vectors. 
\end{proposition}
\begin{proof}
	For contradiction, we assume that the intersection only contains less than $r$ linearly independent vectors. Since $\mathrm{rank}(V)=r$, there exists a vector $v_r\in\mathcal{R}(V)$ that is linearly independent with vectors in the intersection. For any $v^{\prime}$ in the intersection and any $\lambda\neq 0$, the vector $(v^{\prime}+\lambda v_r)\in\mathcal{R}(V)$ is not contained in the intersection, indicating that it is also not in $\mathcal{O}_{s}$. This means that one cannot define $\alpha_{v^{\prime}+\lambda v_r}(i)$ for $i\in \mathbb{Z}_{[1,n]}$ which are all positive (see~(\ref{def-orthant})). However, since $v^{\prime}\in \mathcal{O}_{s}$, there exsit corresponding $\alpha_{v^{\prime}}(i)>0$, for all $i\in\mathbb{Z}_{[1,n]}$, which are all positive. Thus, for sufficiently small $\lambda>0$, the vector $v^{\prime}+\lambda v_r$ is close enough to $v^{\prime}$ that it can be expressed as in~(\ref{def-orthant}) with $\alpha_{v^{\prime}+\lambda v_r}(i)>0$ for all $i\in\mathbb{Z}_{[1,n]}$. This shows that $v^{\prime}+\lambda v_r \in \mathcal{O}_{s}$, which contradicts  our assumption. Hence, the intersection contains $r$ linearly independent vectors.
\end{proof}

\section{Local observability analysis of FNNs}\label{section-observability}
In this section, we reformulate the input-output behavior of FNNs as dynamical systems by treating their weights as states. Using this reformulation, we analyze local observability of several FNN architectures via the observability rank condition. This analysis is essential for establishing the convergence of FNN training when employing MHE in the subsequent section.  We begin in Section~\ref{2-layer FNN} with a specific two-layer FNN, deriving a sufficient condition for the reformulated FNN system to be locally observable using the theory of elementary vectors. This analysis is then extended to general multi-layer FNNs in Section~\ref{section-multi FNN}.

It is well known that, in general, the weights of a neural network that produce a given input-output mapping are not unique \cite{bonassi2022recurrent, hunt1992neural, li2017convergence, albertini1993uniqueness}. This implies that the weights are indistinguishable from input-output data, as illustrated in the following motivating example.

\textit{Example 1}: Consider a two-layer FNN with $1$ input, $3$ hidden nodes, no bias, and $1$ output. The activation functions in hidden nodes are ReLUs and that in the output node is the identity function. Let the ideal weights in the hidden layer be nonzero and denoted by $W = [a, b, c]$, and the output weights are all fixed to $1$.  
In this setting, when trying to estimate the ideal weights there are two possible cases. First, suppose that $a,b$, and $c$
have the same sign. Then, the  corresponding output $y=\sigma(au)+\sigma(bu)+\sigma(cu)$ can only take values of either $0$ or $\abs{(a + b + c)u}$, and therefore the individual weights cannot be determined under any input sequence. The
other case is that only two of the weights $a, b$, and $c$ have the same sign. Then, suppose that $a, b > 0$ and $c < 0$. In this
case, it yields $y = cu$ for any $u < 0$, and $y = (a + b)u$ for $u > 0$. In this case, the weights $a$ and $b$ cannot be uniquely determined. Following similar arguments, it can be observed that any FNN with $1$ input and more than $3$ hidden nodes inevitably has indistinguishable weights.

This example shows that, in general,  the weights of FNNs are not distinguishable, as even a simple FNN architecture fails to satisfy this property. However, we show that in certain circumstances and for specific architectures, the weights of an FNN can be locally distinguishable. Therefore, we first consider a specific class of two-layer FNNs.

\subsection{Analysis of a class of two-layer FNNs}\label{2-layer FNN}
\subsubsection{Problem formulation}
We consider two-layer FNNs, with an $m$-node input layer, an $n$-node hidden layer, and a single  output.\footnote{As specified later, we consider the case where all hidden-to-output weights are fixed to $1$. Consequently, multiple outputs would be redundant, and hence we focus on the single-output architecture in this section. For arbitrary nonzero fixed output weights, each column of the Jacobian matrix in~(\ref{2-eq-H-derivative}) below is scaled by a constant. In this case, considering  multiple output nodes increases the number of rows of the counterpart of (\ref{2-eq-H-derivative}) without changing the number of  columns,  making it easier to have full rank under the same input sequence. } Let $w_{i,j}$ denote the weight from the $i$th node in the input layer to the $j$th node in the hidden layer, and $w_j\coloneqq[w_{1,j}, w_{2,j},\ldots, w_{m,j}]^{\top}, j\in\mathbb{Z}_{[1,n]}$ denote the weights from the whole input layer to the $j$th node of the hidden layer. Then, the weights from the input layer to the hidden layer are denoted by $W\coloneqq[w_1,w_2, \ldots, w_n]\in \mathbb{R}^{m\times n}$, as shown in Fig.~\ref{fig:FNN2}.  We consider the case where all weights from the hidden layer to the output layer are fixed to $1$ and bias weights are omitted, and show that even under this specific architecture, guaranteeing local observability of the reformulated system remains restrictive. The inclusion of output and bias weights as a part of the states to be estimated is addressed in Section~\ref{section-multi FNN}. Following the same setting as in Section~\ref{section-pre-FNN}, the activation functions in the hidden layer and the output layer are ReLUs and the identity function, respectively.
\begin{figure}[!tb]
	\centering
	\begin{tikzpicture}[scale=0.9, 
		auto, 
		every node/.style={transform shape}, 
		every text node part/.style={align=center},
		neuron/.style={circle, draw, minimum size=0.7cm}, 
		connect/.style={draw, -latex},
		layer/.style={align=center},
		midpoint/.style = {coordinate}
		]
		\node[neuron] (I1) at (0, 0.65) {}; 
		\node at ([xshift=-0.8cm] I1) {$u(1)$};
		\node[neuron] (I2) at (0, -0.65) {}; 
		\node at ([xshift=-0.8cm] I2) {$u(m)$};
		
		\node at (0, 0.1) {\vdots}; 
		
		\node[neuron] (H1) at (2, 1.2) {$\sigma$}; 
		\node[neuron] (H2) at (2, 0.15) {$\sigma$}; 
		\node at (2, -0.4) {\vdots}; 
		\node[neuron] (H3) at (2, -1.15) {$\sigma$}; 
		
		\node[neuron] (O1) at (4, 0) {$\scriptstyle \sum$};
		
		\foreach \j in {2,3} {
			\foreach \i in {1,2} {
				\draw[connect] (I\i) -- (H\j);
			}
		}
		\draw[connect] (I1) -- node{$W$}(H1);
		\draw[connect] (I2) -- (H1);
		\foreach \i in {1,2,3} {
			\draw[connect] (H\i) -- (O1);
		}
		\node [midpoint, right of=O1, node distance=1cm] (output) {};
		\draw[connect] (O1) -- node{$y$}(output);
		\node[layer, above of=I1, node distance=1.35cm] (input) {Input \\ layer};
		\node[layer, above of=H1, node distance=0.8cm] (hidden) {Hidden \\ layer};
		\node[layer, above of=O1, node distance=2cm] (output) {Output \\ layer};
	\end{tikzpicture}
	\caption{A specific class of two-layer FNNs}
	\label{fig:FNN2}
\end{figure}
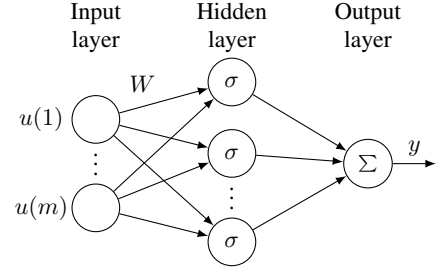

Based on the described FNN setting, the input-output mapping is given by
\begin{align}\label{2-eq:FNN-2}
	y=h(W,u)=\sum_{j=1}^{n}\sigma(w_j^{\top}u).
\end{align} 

In the following, we denote the concatenation of the weights $w_j$, $j\in \mathbb{Z}_{[1,n]}$, by $w\coloneqq (w_1^{\top},w_2^{\top},\ldots,w_n^{\top})^{\top}\in\mathbb{R}^{mn}$. We define the indicator function $\chi(a)$ as
\begin{align}
	\chi(a) =
	\begin{cases}
		1, & a > 0, \\
		0, & a \leq 0.
	\end{cases}
\end{align}
For vector or matrix arguments, the indicator function is applied element-wise. Then, for an input $u_i\in \mathbb{R}^m$, and given that the activation functions $\sigma$ are ReLUs, the partial derivative of the output with respect to the weights is\footnote{ Note that these derivatives are well defined if the arguments of $\sigma$ in~(\ref{2-eq:FNN-2}) are nonzero. The input sequences constructed later in Theorems~\ref{th-local observable} and~\ref{th-PEU} are such that this property holds and hence the derivatives are well defined.\label{footnote-differentiable}} 
\begin{align}\label{2-eq-derivative}
	\frac{\partial y}{\partial w}&=\left[ \frac{\partial h}{\partial w_1}, \ldots ,\frac{\partial h}{\partial w_n}\right]=\left[ u_i^{\top} \chi_{i,1},  \ldots ,u_i^{\top} \chi_{i,n}\right],
\end{align}where 
\begin{align}\label{2-eq-indicator}
	\chi_{i,j}\coloneqq \chi(w_j^{\top}u_i),~j\in \mathbb{Z}_{[1,n]}.
\end{align}

To formally analyze the observability of FNNs, we first represent them as dynamical systems. Specifically, the input-output mapping in~(\ref{2-eq:FNN-2}) admits the representation
\begin{subequations}\label{2-eq-FNN problem}
	\begin{align}
		w_{t+1}&=w_{t},\label{2-eq-FNN state}\\
		y_t&=h(w_{t},u_{t}),\label{2-eq-FNN output}
	\end{align}
\end{subequations}with $t\in \mathbb{Z}_{\geq 0}$, where  $u_t\in \mathbb{R}^{m}$, $w_t\in \mathbb{R}^{nm}$, and $y_t\in \mathbb{R}$ are  the input, the weights (also the states), and the output at time step~$t$.  Note that since the ideal weights of an FNN are constant, the corresponding dynamical system representation (\ref{2-eq-FNN state}) is static.  Variants of the model in~(\ref{2-eq-FNN problem}) have been employed to train neural networks through state estimation methods \cite{bemporad2022recurrent, bonassi2022towards, singhal1988training}. With a slight abuse of notation from~(\ref{2-eq:FNN-2}), we denote the FNN input-output mapping by $h(w_{t},u_{t})$ in subsequent sections, where the state $w_t$ is the concatenation of the columns of $W$ at time step $t$.

Since the output~(\ref{2-eq-FNN output}) is a function solely of the current state and input, independent of any past inputs, the observability mapping of the FNN at state $w$ under the input sequence $u_{[1,N]}$ for some positive integer $N$ is given by 
\begin{align}
	\mathscr{H}_N(w)=[h(w,u_1),h(w,u_2), \ldots, h(w,u_N)]^{\top}.
\end{align}Then, combining with (\ref{2-eq-derivative}) and (\ref{2-eq-indicator}), we obtain the Jacobian matrix of $\mathscr{H}_N(w)$, given by
\begin{align}
	\mathrm{D}\mathscr{H}_{N}(w)&=\begin{bmatrix}
		\frac{\partial h}{\partial w_1}(w,u_1)  &\!\!\cdots\!\! &\frac{\partial h}{\partial w_n}(w,u_1)\\
		\frac{\partial h}{\partial w_1}(w,u_2)  &\!\!\cdots \!\!&\frac{\partial h}{\partial w_n}(w,u_2)\\
		\vdots &\!\!\ddots\!\!&\vdots\\
		\frac{\partial h}{\partial w_1}(w,u_{N})  &\!\!\cdots\!\! &\frac{\partial h}{\partial w_n}(w,u_{N})\\
	\end{bmatrix}\notag\\
	&=\begin{bmatrix}\label{2-eq-H-derivative}
		u_1^{\top} \chi_{1,1}  & \cdots &u_1^{\top} \chi_{1,n}\\
		u_2^{\top} \chi_{2,1}  & \cdots &u_2^{\top} \chi_{2,n}\\
		\vdots&\ddots&\vdots\\
		u_{N}^{\top} \chi_{N,1} & \cdots &u_{N}^{\top} \chi_{N,n}\\
	\end{bmatrix}.
\end{align}

Since $\mathrm{D}\mathscr{H}_N(w)\in \mathbb{R}^{N\times mn}$, satisfaction of the observability rank condition in~(\ref{2-eq-ORC}) requires $N\geq mn$. This Jacobian matrix can be further factorized into the product of a matrix depending on the input sequence and a matrix of indicator functions, as shown below, 
\begin{align}\label{2-eq-dH-factorization}
	\mathrm{D}\mathscr{H}_N(w)=\underbrace{
		\begin{bsmallmatrix}
			u_1^{\top}&~&~&~\\
			~&u_2^{\top}&~&~\\
			~&~&\ddots&~\\
			~&~&~&u_N^{\top}
	\end{bsmallmatrix}}_{T_u}\left(\underbrace{\begin{bsmallmatrix}
			\chi_{1,1} &  \chi_{1,2} & \cdots & \chi_{1,n}\\
			\chi_{2,1} &  \chi_{2,2} & \cdots & \chi_{2,n}\\
			\vdots&\vdots&\ddots&\vdots\\
			\chi_{N,1} &  \chi_{N,2} & \cdots & \chi_{N,n}
	\end{bsmallmatrix}}_{T_{\chi}}\otimes\mathrm{I}_{m}\right),
\end{align}where $T_u\in\mathbb{R}^{N\times Nm}$, $T_{\chi}\in\mathbb{R}^{N\times n}$. Denoting $U\coloneqq[u_1,\ldots,u_N]^{\top}\in \mathbb{R}^{N\times m}$, it follows that $\chi(UW)=[\chi(u_1^{\top}W)^{\top},\cdots,\chi(u_N^{\top}W)^{\top}]^{\top}$. The indicator function corresponding to $u_i^{\top}W$  is given by $\chi(u_i^{\top}W)=\left[ \chi(u_i^{\top}w_1), \ldots ,\chi(u_i^{\top}w_n)\right]$. Recalling the definition of $\chi_{i,j}$ in~(\ref{2-eq-indicator}), and that $\chi(w_j^{\top}u_i)=\chi(u_i^{\top}w_j)$, it yields $\chi(u_i^{\top}W)=\left[ \chi_{i,1}, \ldots ,\chi_{i,n}\right]$. Therefore, we have 
\begin{align}\label{def-Tchi}
	 T_{\chi}=\chi(UW).
\end{align}

The factorization in~(\ref{2-eq-dH-factorization}) enables us to derive conditions for $\rank{\mathrm{D}\mathscr{H}_N(w)}=mn$, i.e., satisfaction of the observability rank condition in~(\ref{2-eq-ORC}). For instance, a necessary condition for $\mathrm{D}\mathscr{H}_N(w)$ to have full column rank is that $T_{\chi}$ has full column rank.

\subsubsection{Sufficient condition for local observability}
Beyond the FNN architecture itself, observability of a particular configuration also depends on the network weights, which determine whether there exists an input sequence that satisfies the observability rank condition (\ref{2-eq-ORC}). In what follows, we demonstrate that specific FNN configurations are  locally observable and  provide a sufficient condition for local observability.

\begin{theorem}\label{th-local observable}
	Consider the FNN in~(\ref{2-eq-FNN problem}) with the weight matrix $W\in\mathbb{R}^{m\times n}$, where $m\leq n$. Suppose $W$ has no zero columns and its row space $\mathcal{R}(W)$ intersects $\ell$ orthants $\mathcal{O}_{s_i}$ in $\mathbb{R}^{n}$ for $i\in \mathbb{Z}_{[1,\ell]}$, with corresponding sign vectors $s_i\in\{-1,1\}^{1\times n}$. Let the sign matrix of $\mathcal{R}(W)$ be $S_{RW}\coloneqq [s_1^{\top},s_2^{\top},\ldots,s_{\ell}^{\top}]^{\top}\in \mathbb{R}^{\ell\times n}$.  Then, the state $w$ is locally observable if 
	\begin{align}\label{3-th-iff-eq}
		\mathrm{rank}(\chi(S_{RW}))=n.
	\end{align}
\end{theorem}
\begin{proof}
Without loss of generality, we prove this theorem for the case where $W$ has full row rank. Notice that if $W$ was row-rank deficient, there would be redundant inputs to the FNN, and an equivalent network could be obtained with reduced inputs and a $W$ of full row rank.  In the following, we show that $w$ is locally observable by establishing the existence of an input sequence $u_{[1,N]}$ with $N=mn$ such that the observability rank condition in~(\ref{2-eq-ORC}) is satisfied at state $w$.

	From (\ref{3-th-iff-eq}), it follows that $\mathcal{R}(W)$ intersects at least $n$ orthants in $\mathbb{R}^{n}$, implying $\ell\geq n$. Let $T_k\in \mathbb{R}^{1\times n}$ for $k\in \mathbb{Z}_{[1,n]}$ denote $n$ linearly independent rows of $\chi(S_{RW})$. Each $T_k$ is the indicator vector of a distinct intersected orthant's sign vector. By Proposition~\ref{3-proposition-orthant and subspace reconstruction} and the fact that $\mathrm{rank}(W)=m$, each intersection of $\mathcal{R}(W)$ with an orthant contains $m$ linearly independent vectors.   Therefore,  there exists a full row rank matrix $C_k\in \mathbb{R}^{m\times n}$ formed by $m$ such linearly independent vectors from the cone $\mathcal{R}(W) \cap \mathcal{O}_{s_{i(k)}}$, where $i(k)\in\mathbb{Z}_{[1,\ell]}$ and $T_k=\chi(s_{i(k)})$. Since all these vectors lie in the same orthant $\mathcal{O}_{s_{i(k)}}$, they share the sign vector $s_{i(k)}$ and hence the same indicator vector $T_k$, yielding $\chi(C_k)=T_k\otimes \mathbf{1}_m$. 
	
	Since $\mathrm{rank}(C_k)=m$ and $\mathcal{R}(C_k)=\mathcal{R}(W)$, it follows that the columns of $C_k^{\top}$ are linear combinations of the columns of $W^{\top}$. Thus, $C_k=T_{C_k}W$ for some nonsingular $T_{C_k}\in \mathbb{R}^{m\times m}$.  We now construct the desired persistently exciting input sequence $u_{[1,N]}$ with $N=nm$. To this end,  define the input blocks for $k\in \mathbb{Z}_{[1,n]}$ as
	\begin{align}
		 	U_k\!\coloneqq\!\begin{bmatrix}
		 	u_{(k-1)m+1}, u_{(k-1)m+2}, \cdots, u_{km}
		 \end{bmatrix}^{\top},\label{3-eq-Hk}
	\end{align} and set $U_k=C_kW^{\dagger}$, for $k\in \mathbb{Z}_{[1,n]}$. The block matrix $U_k$ is nonsingular because $U_k=T_{C_k}WW^{\dagger}=T_{C_k}$, where $W^{\dagger}=W^{\top}(WW^{\top})^{-1}$ is the Moore-Penrose inverse due to the fact that $W$ has full row rank.

	We now prove the theorem by showing that  the observability rank condition $\textrm{rank}(\mathrm{D}\mathscr{H}_N(w))=nm$ holds under the constructed input sequence $u_{[1,N]}$. To this end, we assume for contradiction that $\textrm{rank}(\mathrm{D}\mathscr{H}_N(w))<nm$. Then, there exist some coefficients $c_{1}, c_{2},\ldots, c_{mn}$, not all zero, such that
	\begin{align}
		\begin{bmatrix}
			u_1^{\top} \chi_{1,1} & u_1^{\top} \chi_{1,2} & \cdots &u_1^{\top} \chi_{1,n}\\
			u_2^{\top} \chi_{2,1} & u_2^{\top} \chi_{2,2} & \cdots &u_2^{\top} \chi_{2,n}\\
			\vdots&\vdots&\ddots&\vdots\\
			u_{N}^{\top} \chi_{N,1} & u_{N}^{\top} \chi_{N,2} & \cdots &u_{N}^{\top} \chi_{N,n}\\
		\end{bmatrix}\begin{bmatrix}
			c_1\\c_2\\ \vdots \\ c_{mn}
		\end{bmatrix}=\mathbf{0}.\label{4-eq-th-1}
	\end{align}By taking $m$-row subsets of~(\ref{4-eq-th-1}), we have
	\begin{align}
		\begin{bsmallmatrix}\label{4-eq-th-subset rows}
			u_{(k-1)m+1}^{\top} \chi_{(k-1)m+1,1}  & \cdots &u_{(k-1)m+1}^{\top} \chi_{(k-1)m+1,n}\\
			u_{(k-1)m+2}^{\top} \chi_{(k-1)m+2,1}  & \cdots &u_{(k-1)m+2}^{\top} \chi_{(k-1)m+2,n}\\
			\vdots&\ddots&\vdots\\
			u_{km}^{\top} \chi_{km,1} &  \cdots &u_{km}^{\top} \chi_{km,n}\\
		\end{bsmallmatrix}\begin{bsmallmatrix}
			c_1\\c_2\\ \vdots \\ c_{mn}
		\end{bsmallmatrix}=\mathbf{0},
	\end{align}where $ k\in\mathbb{Z}_{[1,n]}$. 
	
	Recalling the definition of $T_{\chi}$ in (\ref{def-Tchi}) and that $U_k=C_kW^{\dagger}$,  it is clear that $U_kW=T_{C_k}W=C_k$ and thus $T_{\chi}=\chi(UW)=\chi(C)$, where $C\coloneqq [C_1^{\top},\cdots, C_n^{\top}]^{\top}$. Combining this with $\chi(C_k)=T_k\otimes \mathbf{1}_m$, we obtain
	\begin{align}\label{4-eq-th-subset indicator}
		\begin{bmatrix}
			\chi_{(k-1)m+1,1} &   \cdots & \chi_{(k-1)m+1,n}\\
			\chi_{(k-1)m+2,1} &  \cdots & \chi_{(k-1)m+2,n}\\
			\vdots&\ddots&\vdots\\
			\chi_{km,1} &   \cdots & \chi_{km,n}
		\end{bmatrix}=T_{k}\otimes\mathbf{1}_m,
	\end{align}where $ k\in\mathbb{Z}_{[1,n]}$. Since  $T_{k}\in\mathbb{R}^{1\times n}$ is an indicator vector with entries in $\{0,1\}$, each column of the matrix on the left-hand side of~(\ref{4-eq-th-subset indicator}) is either an all-ones vector or a zero vector, with at least one column being an all-ones vector. As a result, the leftmost matrix in~(\ref{4-eq-th-subset rows}) consists of blocks, each of which equals either $U_k$ or $\mathbf{0}_{m\times n}$.
	
By assumption, (\ref{4-eq-th-1}) holds with coefficients $c_1,\ldots,c_{mn}$ not all zero, implying that the columns of the leftmost matrix of (\ref{4-eq-th-1}), 
and hence also of the left-most matrix in~(\ref{4-eq-th-subset rows}), are linearly independent. Since $U_k$ is nonsingular, this can only be the case if, for every every $j\in \mathbb{Z}_{[1,m]}$, the columns with indices $j,m+j,\ldots,(n-1)m+j$ are linearly dependent. That is, we have
	\begin{align}
		\sum_{l=1}^{n}u_{(k-1)m+i}(j)\chi_{(k-1)m+i,l}c_{(l-1)m+j}=0,
	\end{align}where  $k\in \mathbb{Z}_{[1,n]}$, $i\in\mathbb{Z}_{[1,m]}$, and  $j\in\mathbb{Z}_{[1,m]}$. This is equivalent to
	\begin{align}\label{4-eq-coefficient0}
		\mathrm{diag}(u_{(k-1)m+i})\underbrace{\begin{bsmallmatrix}
				c_1&c_{m\!+\!1}&\cdots&c_{m(n\!-\!1)+1}\\
				c_2&c_{m\!+\!2}&\cdots&c_{m(n\!-\!1)+2}\\
				\vdots&\vdots&\ddots&\vdots\\
				c_{m}&c_{2m}&\cdots&c_{nm}\\
		\end{bsmallmatrix}}_{\Gamma}T_{k}^{\top}\!=\!\mathbf{0},
	\end{align}
	for all $i\in \mathbb{Z}_{[1,m]}$ and $k\in\mathbb{Z}_{[1,n]}$, which can be represented as 
	\begin{align}\label{22}
		\mathrm{diag}(\Gamma T_{k}^{\top})u_{(k-1)m+i}=\mathbf{0},
	\end{align}for all $i\in \mathbb{Z}_{[1,m]}$ and $k\in\mathbb{Z}_{[1,n]}$. Collecting these equations for all $i$ yields
	\begin{align}
		\mathrm{diag}(\Gamma T_{k}^{\top}) U_k^{\top}=\mathbf{0}, k\in\mathbb{Z}_{[1,n]}.
	\end{align}
	Since $U_k$ is nonsingular, this results in
	\begin{align}\label{4-eq-coefficient1}
		\Gamma T_{k}^{\top}=\mathbf{0}, k\in\mathbb{Z}_{[1,n]}.
	\end{align}Then, combining~(\ref{4-eq-coefficient1}) for all $k$ yields
	\begin{align}\label{4-eq-coefficient2}
		\Gamma\begin{bmatrix}
			T_1^{\top},T_2^{\top},\ldots,T_n^{\top}
		\end{bmatrix}=\mathbf{0}.
	\end{align}Recalling that the vectors $T_k$ for $k\in\mathbb{Z}_{[1,n]}$ are linearly independent, equation (\ref{4-eq-coefficient2}) implies that all coefficients $c_{1},\cdots, c_{mn}$ must be zero. However, this contradicts our initial assumption in~(\ref{4-eq-th-1}) that not all $c_i$ are zero. Therefore, $\textrm{rank}(\mathrm{D}\mathscr{H}_N(w))=nm$, which means the observability rank condition is satisfied at state $w$ for the constructed input sequence, and the state $w$ is locally observable.
\end{proof}

As mentioned in the proof above, the input-output mapping of an FNN with a row-rank deficient weight matrix can always be presented by  an equivalent FNN whose weight matrix has lower dimension and full row rank. Therefore, also in the subsequent analysis, we restrict our attention to the case where $W$ has full row rank.

In our preliminary work \cite{yang2025local}, we considered FNNs with equal numbers of inputs and hidden nodes. For this architecture, local observability was established under the condition that $W$ is nonsingular. This condition implies that the row space $\mathcal{R}(W)$, which is of dimension $n$, intersects every orthant in $\mathbb{R}^n$. Therefore, the corresponding indicator matrix $\chi(S_{RW})$ satisfies condition~(\ref{3-th-iff-eq}), rendering that state $w$ is locally observable by Theorem~\ref{th-local observable}. Thus, Theorem~\ref{th-local observable} covers the results of \cite{yang2025local} as a special case.

\begin{remark}
To verify condition (\ref{3-th-iff-eq})  for a given weight matrix $W$, the sign vectors $s_i$ for $i\in\mathbb{Z}_{[1,\ell]}$ can be obtained using the package in~\cite{aichmayr2024sagemath} (see also the discussion above Proposition~\ref{3-proposition-orthant and subspace reconstruction}).
\end{remark}

\begin{remark}\label{rem-lower bound}
	Theorem~\ref{th-local observable} establishes a sufficient condition for identifying a locally observable state based on the observability rank condition. The minimal input sequence length for this condition to hold is $N=mn$, since the matrix $\mathrm{D}\mathscr{H}_N(w)$ has $mn$ columns. Therefore, for any locally observable state satisfying the condition in Theorem~\ref{th-local observable}, a minimal input sequence of length $mn$ is sufficient for excitation as proved above, even though longer persistently exciting input sequences (with $N>mn$) are possible.
\end{remark}
\begin{remark}
	Theorem~\ref{th-local observable} addresses the case $m\leq n$. This covers the typical FNN  architecture in practice, where the hidden layer has at least as many nodes as the input layer \cite{cybenko1989approximation}. If instead $m>n$, the results of Theorem~\ref{th-local observable} still hold. A similar proof can be carried out by only adapting the construction of input blocks in~(\ref{3-eq-Hk}), where a full row rank $W$ is considered. In particular, without this requirement, a nonsingular $U_k$ satisfying $U_kW=C_k$  exists as well, since $C_k$ is constructed such that $\mathcal{R}(C_k)=\mathcal{R}(W)$.  Moreover, in this case the architecture can be reduced to an equivalent network with $m \leq n$, which amounts to removing redundant inputs.  For these reasons, we only focus on the case $m\leq n$ in the following sections. 
\end{remark}
\subsubsection{Persistently exciting input design}
In the proof of Theorem~\ref{th-local observable}, it was shown how to design persistently exciting input sequences for a given locally observable state~$w$ in the FNN configuration~(\ref{2-eq:FNN-2}). This method can be used for design purposes and is summarized in the following result.

\begin{theorem}\label{th-PEU}
	Consider the reformulated FNN dynamical system in~(\ref{2-eq-FNN problem}). Assume a full row rank weight matrix $W$ satisfies condition~(\ref{3-th-iff-eq}) in Theorem~\ref{th-local observable}. For any $C=[C_1^{\top}, C_2^{\top},\ldots, C_n^{\top}]^{\top}$, $C_k\in\mathbb{R}^{m\times n}$, $k\in\mathbb{Z}_{[1,n]}$, satisfying
	\begin{align}
		&\mathrm{rank}(C_k)=m, k\in\mathbb{Z}_{[1,n]},\label{4-th-c1}\\
		&C_k=C_kW^{\dagger}W, k\in\mathbb{Z}_{[1,n]}, \label{4-th-c2}\\
		&\chi(C)=T\otimes \mathbf{1}_{m},\label{4-th-c3}
	\end{align}where $T\in\mathbb{R}^{n\times n}$ is a nonsingular matrix with binary entries (each element is either $1$ or $0$), the input sequence $u_{[1,N]}$ derived by
	\begin{align}
		U\coloneqq [u_1, u_2,\ldots,u_N]^{\top}= CW^{\dagger},\label{4-th-c4}
	\end{align} is persistently exciting.
\end{theorem} 
\begin{proof}
Condition (\ref{4-th-c2}) guarantees $\mathcal{R}(C_k)\subseteq\mathcal{R}(W)$. Together with~(\ref{4-th-c1}), it follows that the columns of $C_k^{\top}$  span  $\mathcal{R}(W)$, and $C_k=T_{C_k}W$ holds for some nonsingular $T_{C_k}\in \mathbb{R}^{m\times m}$. Consequently, $UW=C$ holds and each $U_k=C_kW^{\dagger}$ is nonsingular for $k\in\mathbb{Z}_{[1,n]}$, since $W$ has full row rank and $U_k=T_{C_k}WW^{\dagger}=T_{C_k}$.

Condition (\ref{4-th-c3}) ensures that  all columns of $C_k^{\top}$ share the same indicator vector, which is exactly the condition~(\ref{4-eq-th-subset indicator}). This implies that  all columns of $C_k^{\top}$ lie in the same orthant. Furthermore, since $T$ is nonsingular, the columns of $C_k^{\top}$ and $C_{k^{\prime}}$ for $k\neq k^{\prime}$ lie in different orthants. Thus, the conditions in Theorem~\ref{th-PEU} fulfill the requirements in the proof of Theorem~\ref{th-local observable} (see the paragraphs below~(\ref{3-eq-Hk})) to construct an input sequence satisfying the observability rank condition for any state $w$ meeting~(\ref{3-th-iff-eq}). Therefore, any input sequence designed via~(\ref{4-th-c4}) is persistently exciting. 
\end{proof}

As shown in the proof of Theorem~\ref{th-PEU}, each $C_k$ satisfying~(\ref{4-th-c1})-(\ref{4-th-c3}) is constructed from vectors within the cone $\mathcal{R}(W)  \cap O_{s_i}$ for some orthant $O_{s_i}$. Thus, there are infinitely many admissible choices for $C_k$, resulting in different persistently exciting input sequences. Note that $C_k$ is constructed from vectors in open orthants. Hence, it contains no zero entries and all entries of $UW$ are nonzero, ensuring that  $\mathrm{D}\mathscr{H}_N(w)$ is well defined (compare Footnote~\ref{footnote-differentiable}).

\subsubsection{Locally observable neighborhood}
The persistently exciting input  design method  in Theorem~\ref{th-PEU} requires prior knowledge of a specific state whose weight matrix $W$ satisfies condition~(\ref{3-th-iff-eq}). This requirement is restrictive, especially when employing MHE to estimate an FNN's optimal state, which is unknown a priori.  In the following theorem, we show that if an input sequence is constructed to render a state $w$ distinguishable, then all states within a neighborhood of $w$ become also distinguishable under the same input. 
\begin{theorem}\label{th-observable set}
	Consider a locally observable state $w$ of system~(\ref{2-eq-FNN problem}) with a weight matrix $W$, and a persistently exciting input sequence $u_{[1,N]}$. A state $w^{\prime}$, associated with a weight matrix $W^{\prime}$, is also excited by the same input and  locally observable if there exists a matrix $K\in\mathbb{R}^{N\times n}$ whose entries are all in $(-1,+\infty)$, such that 
	\begin{align}\label{th-observable set-c1}
		U\delta=K\circ UW
	\end{align}holds, where $U\coloneqq[u_1,u_2,\ldots,u_N]^{\top}$ and $\delta=W^{\prime}-W$.
\end{theorem} 
\begin{proof}
	We show that, under the conditions of the theorem, the observability rank condition in (\ref{2-eq-ORC}) also holds at the state $w^{\prime}$ under the same input $u_{[1,N]}$. From condition~(\ref{th-observable set-c1}), it follows that each entry of $U\delta$ equals the corresponding entry of $UW$ scaled by a constant greater than $-1$. This ensures that the signs of the entries in $UW+U\delta$ are the same as those in $UW$. Consequently, $\chi(UW+U\delta)=\chi(UW)$,  resulting in the same indicator matrix $T_{\chi}$ according to~(\ref{def-Tchi}) for $w^{\prime}$ and $w$. 
	
	Since the input sequence is persistently exciting for the state $w$, $\mathrm{D}\mathscr{H}_N(w)$ has full column rank. Thus, by the fact that $T_{\chi}$ remains unchanged for $w^{\prime}$ under the same input matrix $T_u$, $\mathrm{D}\mathscr{H}_N(w^{\prime})$ also remains identical to $\mathrm{D}\mathscr{H}_N(w)$ and has full column rank (see~(\ref{2-eq-dH-factorization})). Therefore, the unknown state $w^{\prime}$ is both locally observable and excited by the input $u_{[1,N]}$.
\end{proof}

Theorem~\ref{th-observable set} constrains the distance between a locally observable state $w$ and other unknown states $w^{\prime}$ that can also be excited by the same input. Note that the satisfaction of condition~(\ref{th-observable set-c1}) admits a neighborhood of the state $w$ when all entries of $UW$ are nonzero. This follows because, for sufficiently small $\delta$, condition~(\ref{th-observable set-c1}) can always be satisfied by appropriately choosing a matrix $K$, as the allowed interval for the entries of $K$ contains $0$ in its interior.   

The following definition and subsequent result will be useful later in Section~\ref{section-MHE}. 
\begin{definition}[locally observable neighborhood of $w$]\label{def-observable set}
	 For a locally observable state $w$ with weight matrix $W$, and a persistently exciting input $U$, the locally observable neighborhood $\mathbb{W}_0$ of $w$ is defined as the set of all states $w^{\prime}$ that are also excited by $U$ and  satisfy:
	 \begin{align*}
	 	\mathbb{W}_0\!=\!\{\mathrm{vec}(W^{\prime})\mid U(W^{\prime}\!-\!W)\!=\!K\circ UW, K_{i,j}\!\in\!(-1,+\infty)\},
	 \end{align*}where $\mathrm{vec}(\cdot)$ denotes the vectorization operator that stacks the columns of a matrix into a vector.
\end{definition}

In the following, we show that $\mathbb{W}_0$ is a convex set, which plays an important role in  the design and analysis of an MHE-based training algorithm for FNNs.
\begin{proposition}\label{prop-convex}
For a locally observable state $w$, the locally observable neighborhood $\mathbb{W}_0$ of $w$ is a convex set.
\end{proposition}
\begin{proof}
	Consider states $w_1\in\mathbb{W}_0$ and $w_2\in\mathbb{W}_0$, with corresponding weight matrices $W_1$ and $W_2$, respectively.  For any scalar $\alpha\in[0,1]$, we have a $W_3=\alpha W_1+(1-\alpha)W_2$, with associated state $w_3$. By~(\ref{th-observable set-c1}), it follows that $U(W_1-W)=K_1\circ UW$ and $U(W_2-W)=K_2\circ UW$, where $K_1$ and $K_2$ satisfy the condition in Theorem~\ref{th-observable set}. Therefore, $UW_3$ is given by
		\begin{align}
		UW_3&=U(\alpha W_1+(1-\alpha)W_2)\notag\\
		&=\alpha UW_1+(1-\alpha)UW_2\notag\\
		&=UW+\alpha K_1\circ UW+(1-\alpha)K_2\circ UW\notag\\
		&= UW+K_3\circ U W,
	\end{align}where $(K_3)_{i,j}\in(-1,+\infty)$, since $(-1,+\infty)$ is a convex set and $\alpha (K_1)_{i,j}+(1-\alpha)(K_2)_{i,j}=(K_3)_{i,j}$ for all $i\in\mathbb{Z}_{[0,N]}$, $j\in\mathbb{Z}_{[0,n]}$. This result yields that $w_3\in \mathbb{W}_0$ and $\mathbb{W}_0$ is a convex set.
\end{proof}
\subsection{Analysis of general FNNs}\label{section-multi FNN}
In the following, we extend our analysis to more general FNN architectures. We first consider two-layer FNNs with bias weights, establishing local observability of the weights under additional constraints. We then extend the analysis to multi-layer FNNs, showing that local observability properties no longer hold in general. 
\subsubsection{Two-layer FNNs with bias weights} 
Consider two-layer FNNs with an input layer of $m$ nodes, a hidden layer of $n$ nodes, and a single output node as shown in Fig.~\ref{fig:FNN2}. In addition, $n$ bias weights  $b_1, b_2,\cdots,b_n$ are added to the  hidden nodes. 

With a slight abuse of notation, we reuse the symbol $w$ in this section to denote the state of the reformulated FNN system,  $w\coloneqq[w_1^{\top},\ldots,w_n^{\top},b_1,\ldots,b_n]^{\top}\in\mathbb{R}^{(m+1)n}$, which now includes the bias weights. Then, the input-output mapping is given by
\begin{align}\label{plant-FNN-bias}
	y=h(w,u)=\sum_{j=1}^{n}\sigma(w_j^{\top}u+b_j).
\end{align}Analogous to (\ref{2-eq-derivative}), (\ref{2-eq-indicator}), and~(\ref{2-eq-H-derivative}), we obtain the Jacobian of the observability mapping under an input sequence $u_{[1,N]}$, given by 
\begin{align}
	&\mathrm{D}\mathscr{H}_{N}(w)\notag\\
	&=\begin{bsmallmatrix}
		\frac{\partial h}{\partial w_1}(w,u_1) &\cdots &\frac{\partial h}{\partial w_n}(w,u_1)&\frac{\partial h}{\partial b_1}(w,u_{1})&\cdots&\frac{\partial h}{\partial b_n}(w,u_{1})\\
		\frac{\partial h}{\partial w_1}(w,u_2) &\cdots &\frac{\partial h}{\partial w_n}(w,u_2)&\frac{\partial h}{\partial b_1}(w,u_{2})&\cdots&\frac{\partial h}{\partial b_n}(w,u_{2})\\
		\vdots &\ddots&\vdots&\vdots&\ddots&\vdots\\
		\frac{\partial h}{\partial w_1}(w,u_{N}) &\cdots &\frac{\partial h}{\partial w_n}(w,u_{N})&\frac{\partial h}{\partial b_1}(w,u_{N})&\cdots&\frac{\partial h}{\partial b_n}(w,u_{N})\\
	\end{bsmallmatrix}\notag\\
	&=\begin{bmatrix}\label{dH-b}
		u_1^{\top} \chi_{1,1}  & \cdots &u_1^{\top} \chi_{1,n}&\chi_{1,1} &\cdots&\chi_{1,n}\\
		u_2^{\top} \chi_{2,1}  & \cdots &u_2^{\top} \chi_{2,n}&\chi_{2,1} &\cdots&\chi_{2,n}\\
		\vdots&\ddots&\vdots&\vdots&\ddots&\vdots\\
		u_{N}^{\top} \chi_{N,1}  & \cdots &u_{N}^{\top} \chi_{N,n}&\chi_{N,1} &\cdots&\chi_{N,n}\\
	\end{bmatrix}.
\end{align} 
Note that $\chi_{i,j}=\chi(u_i^{\top}w_j+b_j)$, which differs slightly from the definition given in~(\ref{2-eq-indicator}) in Section~\ref{2-layer FNN}. The corresponding weight matrix for state $w$ is now denoted by
\begin{align}
	W\coloneqq\begin{bmatrix}
		W^1\\b^{\top}
	\end{bmatrix}\in\mathbb{R}^{(m+1)\times n},
\end{align}where $W^1=[w_1, \ldots, w_n]\in \mathbb{R}^{m\times n}$, $b=[b_1, \ldots, b_n]^{\top}\in\mathbb{R}^{n}$. Since the bias weights can be interpreted as an additional input node with constant value $1$, and following the analysis in the proof of Theorem~\ref{th-local observable}, we consider only weight matrices $W$ with full row rank.

In the following lemma, we present some results on Greville's method that are obtained by adapting those of~\cite{greville1960some} using the fact that $(A^{\dagger})^{\top} = (A^{\top})^{\dagger}$. This provides a useful expression for the Moore-Penrose inverse that will be exploited in Theorem~\ref{th-bias}. 
	\begin{lemma}[\!\!\cite{greville1960some}]\label{lem-greville}
		Consider a matrix $A_{k}\!\coloneqq\! \begin{bsmallmatrix}
				 A_{k-1}\\a_k
		\end{bsmallmatrix}$, where $a_k$ is the $k$th row of $A_{k}$. Then, the Moore-Penrose inverse of $A_k$ can be presented as
		\begin{align}
			A_k^{\dagger}=\begin{bmatrix}
				A_{k-1}^{\dagger}-b_kd_k,
				b_k
			\end{bmatrix},
		\end{align}	where
		\begin{align}
			d_k&=a_kA_{k-1}^{\dagger},\\
			c_k&=a_k-d_kA_{k-1},\\
			b_k&=\begin{cases}
				c_k^{\dagger},&\text{if}\quad c_k \neq0,\\
				(1+d_kd_k^{\top})^{-1}d_kA_{k-1}^{\dagger}, &\text{if}\quad c_k =0.
			\end{cases}
		\end{align}
	\end{lemma}
	
In the following, we extend the results of Theorems~\ref{th-local observable} and~\ref{th-PEU} to FNN architectures that incorporate bias weights.

\begin{theorem}\label{th-bias}
	Consider the FNN in~(\ref{plant-FNN-bias}) with weight matrix $W\in\mathbb{R}^{(m+1)\times n}$ that has full row rank and no zero columns. Assume the affine subspace $\mathcal{R}(W^1)+b$ intersects $\ell$ open orthants  $\mathcal{O}_{s_i}$ in $\mathbb{R}^{n}$ for $i\in \mathbb{Z}_{[1,\ell]}$, with corresponding sign vectors $s_i\in\{-1,1\}^{ 1\times n}$. Let the sign matrix of $\mathcal{R}(W^1)+b$ be $S_{RWb}\coloneqq [s_1^{\top},s_2^{\top},\ldots,s_{\ell}^{\top}]^{\top}\in \mathbb{R}^{\ell\times n}$.  The state $w$ is locally observable if the corresponding weight matrix satisfies
	\begin{align}\label{3-th-iff-eq-bias}
		\mathrm{rank} (\chi(S_{RWb}))=n.
	\end{align}Moreover, for any $C=[C_1^{\top},\ldots, C_n^{\top}]^{\top}$, $C_k\in\mathbb{R}^{(m+1)\times n}$, $k\in\mathbb{Z}_{[1,n]}$, satisfying
	\begin{align}
	&\mathrm{rank}(C_k)=m+1, k\in\mathbb{Z}_{[1,n]},\label{3-th-c1}\\
	&C_k=C_kW^{\dagger}W, k\in\mathbb{Z}_{[1,n]}, \label{3-th-c2}\\
	&\chi(C)=T\otimes \mathbf{1}_{m+1},\label{3-th-c3}\\
	&C_k(b^{\top}-b^{\top}(W^1)^{\dagger}(W^1))^{\dagger}=\mathbf{1}_{m+1},  k\in\mathbb{Z}_{[1,n]}, \label{3-th-c4}
	\end{align}where $T\in\mathbb{R}^{n\times n}$ is a nonsingular matrix with binary entries (each element is either $1$ or $0$), the input sequence $u_{[1,N]}$ derived by
	\begin{align}
	U\coloneqq [u_1, u_2,\ldots,u_N]^{\top}= C(W^1)^{\dagger}-\mathbf{1}_{N}b^{\top}(W^1)^{\dagger},\label{3-th-c5}
	\end{align} is persistently exciting.
\end{theorem}
\begin{proof}
We first show the existence of a matrix $C$ satisfying conditions~(\ref{3-th-c1})-(\ref{3-th-c4}) when condition~(\ref{3-th-iff-eq-bias}) holds. Since $W$ has full row rank, each intersection $(\mathcal{R}(W^1)+b) \cap \mathcal{O}_{s_{i}}$ contains $m+1$ linearly independent vectors, following a similar analysis as in Proposition~\ref{3-proposition-orthant and subspace reconstruction}. Therefore, conditions~(\ref{3-th-c1})-(\ref{3-th-c3}) are feasible for a matrix $C_k$ constructed from such $m+1$ vectors in the intersection, i.e., the columns of $C_k^{\top}$ correspond to such $m+1$ vectors. Since the columns of $C_k^{\top}$ lie in the affine subspace $\mathcal{R}(W^1)+b$ and condition~(\ref{3-th-c1}) holds, we have 
	\begin{align}
		C_k=[T_{C_k},\mathbf{1}_{m+1}]\begin{bsmallmatrix}
			W^1\\b^{\top}
		\end{bsmallmatrix} \label{ck}
	\end{align} with a nonsingular matrix $[T_{C_k},\mathbf{1}_{m+1}]$. From Lemma~\ref{lem-greville} and the fact that $W$ has full row rank,  it follows that
	\begin{align}\label{eq-th3-ckw}
	C_kW^{\dagger}=C_k\begin{bsmallmatrix}
		W^1\\b^{\top}
	\end{bsmallmatrix}^{\dagger}=C_k[(W^1)^{\dagger}-\alpha\beta,\alpha],
	\end{align} where
	$\beta=b^{\top}(W^1)^{\dagger}$, $\alpha=(b^{\top}-\beta W^1)^{\dagger}$. From~(\ref{ck}), it follows that 
	$C_kW^{\dagger} = [T_{C_k},\mathbf{1}_{m+1}]$. Thus, we obtain
	\begin{align}\label{eq-th3-ckalpha}
	C_k(b^{\top}-b^{\top}(W^1)^{\dagger}W^1)^{\dagger}=C_k\alpha=\mathbf{1}_{m+1},
\end{align}and therefore condition~(\ref{3-th-c4}) holds. 	
	
We continue this proof by showing that under the designed input sequence $u_{[1,N]}$ by (\ref{3-th-c5}), the observability rank condition holds at the state $w$, i.e., $\mathrm{D}\mathscr{H}_{N}(w)$ in (\ref{dH-b}) has full column rank $(m+1)n$.

With a slight abuse of notation, we reuse the symbol $U_k$ and define $U_k\coloneqq[u_{(k-1)(m+1)+1},  \cdots, u_{k(m+1)}]^{\top}\in\mathbb{R}^{(m+1)\times m}$, for $k\in \mathbb{Z}_{[1,n]}$.  Accordingly, the input design method in~(\ref{3-th-c5}) implies that  $U_k=C_k(W^1)^{\dagger}-\mathbf{1}_{m+1}b^{\top}(W^1)^{\dagger}$ for $k\in\mathbb{Z}_{[1,n]}$. By~(\ref{eq-th3-ckalpha}), it follows that
\begin{align}
	 U_k=C_k(W^1)^{\dagger}-C_k\alpha\beta.
\end{align} Combining with~(\ref{eq-th3-ckw}) and~(\ref{eq-th3-ckalpha}), we obtain $C_kW^{\dagger}=[U_k,\mathbf{1}_{m+1}]$. Since $C_k=C_kW^{\dagger}W$ in~(\ref{3-th-c2}), we have $[U_k,\mathbf{1}_{m+1}]W=C_k$, and hence $\chi([U,\mathbf{1}_N]W)=\chi(C)$. Since $W$ and $C_k$ have full row rank, it follows that $\rank{[U_k,\mathbf{1}_{m+1}]}=m+1$, i.e., $[U_k,\mathbf{1}_{m+1}]$ is nonsingular. Together with~(\ref{3-th-c3}), we conclude that $\mathrm{D}\mathscr{H}_{N}(w)$ has full column rank, following an analysis similar to that in the proof of Theorem~\ref{th-local observable} (see the paragraphs below~(\ref{3-eq-Hk})). Therefore, the constructed input by~(\ref{3-th-c5}) is persistently exciting and the state $w$ is locally observable.
\end{proof}

Note that an affine subspace $\mathcal{R}(W^1)+b$  rather than a row space $\mathcal{R}(W)$ is considered, since an additional constraint (\ref{3-th-c4})  is imposed on $C_k$ compared with the design procedures for $C_k$ in Theorem~\ref{th-PEU}. As shown in the proof of Theorem~\ref{th-bias}, constraint~(\ref{3-th-c4}) ensures $C_k=[U_k,\mathbf{1}_{m+1}]W=U_kW^1+\mathbf{1}_{m+1}b^{\top}$, which implies that each column of $C_k^{\top}$  lies in the affine subspace $\mathcal{R}(W^1)+b$.  This distinction arises from the presence of bias weights, which can be equivalently regarded as introducing an additional input node with a constant value of $1$ in the FNN. This fixed-value input node reduces the degrees of freedom in selecting $C_k$, as it must satisfy the constraint $[U_k,\mathbf{1}_{m+1}]W=C_k$. 

\subsubsection{Multi-layer FNNs}
In this section, we analyze the local observability properties  of general multi-layer FNNs. We show that, in general, multi-layer FNNs are not locally observable regardless of the inputs used to collect training data. Although this analysis also holds for deeper FNNs, for the sake of simplicity, the following results are presented for the two-layer network as shown in  Fig.~\ref{fig:FNN2}, but with variable output weights. 

  Let $w^k_{i,j} (k=1,2)$ denote the weight from the $i$th node in layer $k-1$  to the $j$th node in layer $k$. Define $w^k_j\coloneqq[w^k_{1,j}, \ldots, w^k_{m,j}]^{\top}$ as the vector of weights from layer $k-1$ to the $j$th node of layer  $k$. Then, the weights from layer $k-1$ to  layer $k$ are denoted by $W^k\coloneqq[w^k_1, \ldots, w^k_n]$. Therefore, for the considered FNN, we have
\begin{align*}
	W^1&\!=\![w^1_1,\ldots,w^1_n], ~w^1_i\!=\![w^1_{1,i},\ldots,w^1_{m,i}]^{\top}\in\mathbb{R}^{m}, ~i\in \mathbb{Z}_{[1,n]}\\
	W^2&=[w^2_1,\ldots,w^2_n]^{\top}, ~w^2_i\in \mathbb{R}, ~i\in\mathbb{Z}_{[1,n]},\\
	b&=[b_1,\ldots,b_n]^{\top}, ~b_i\in\mathbb{R}, ~i\in\mathbb{Z}_{[1,n]},
\end{align*} where $w^1_i$ is a vector and $w^2_i$ is a scalar.
Again, we reformulate this FNN into a dynamical system of the form in~(\ref{2-eq-FNN problem}), and denote the state of the system by $w\coloneqq[(w^1_1)^{\top},\ldots,(w^1_n)^{\top},b^{\top},w^2_1,\ldots,w^2_n]^{\top}$. Then, the input-output mapping is given by 
\begin{align}
	y=h(w,u)=\sigma(u^{\top}W^1+b)W^2,
\end{align}where $u\in\mathbb{R}^{m}$. The Jacobian matrix of the observability mapping under an input sequence $u_{[1,N]}$ is given by (\ref{full-dH}) (see next page).
\begin{figure*}[!t]
	\begin{equation}
\begin{aligned}
	&\mathrm{D}\mathscr{H}_{N}(w)=\begin{bmatrix}
		\frac{\partial h}{\partial w^1_1}(w,u_1) &\cdots &\frac{\partial h}{\partial w^1_n}(w,u_1)&\frac{\partial h}{\partial b_1}(w,u_{1})&\cdots&\frac{\partial h}{\partial b_n}(w,u_{1})&\frac{\partial h}{\partial w^2_1}(w,u_{1})&\cdots&\frac{\partial h}{\partial w^2_n}(w,u_{1})\\
		\frac{\partial h}{\partial w^1_1}(w,u_2) &\cdots &\frac{\partial h}{\partial w^1_n}(w,u_2)&\frac{\partial h}{\partial b_1}(w,u_{2})&\cdots&\frac{\partial h}{\partial b_n}(w,u_{2})&\frac{\partial h}{\partial w^2_1}(w,u_{2})&\cdots&\frac{\partial h}{\partial w^2_n}(w,u_{2})\\
		\vdots &\ddots&\vdots&\vdots&\ddots&\vdots&\vdots&\ddots&\vdots\\
		\frac{\partial h}{\partial w^1_1}(w,u_{N}) &\cdots &\frac{\partial h}{\partial w^1_n}(w,u_{N})&\frac{\partial h}{\partial b_1}(w,u_{N})&\cdots&\frac{\partial h}{\partial b_n}(w,u_{N})&\frac{\partial h}{\partial w^2_1}(w,u_{N})&\cdots&\frac{\partial h}{\partial w^2_n}(w,u_{N})\\
	\end{bmatrix}\\
	&=\begin{bmatrix}\label{full-dH}
		u_1^{\top} w^2_1\chi_{1,1} &  \cdots &u_1^{\top}w^2_n \chi_{1,n}&w^2_1\chi_{1,1} &\cdots&w^2_n\chi_{1,n}&(u_1^{\top}w^1_1+b_1)\chi_{1,1} &\cdots&(u_1^{\top}w^1_n+b_n)\chi_{1,n}\\
		u_2^{\top} w^2_1\chi_{2,1} &  \cdots &u_2^{\top} w^2_n\chi_{2,n}&w^2_1\chi_{2,1} &\cdots&w^2_n\chi_{2,n}&(u_2^{\top}w^1_1+b_1)\chi_{2,1} &\cdots&(u_2^{\top}w^1_n+b_n)\chi_{2,n}\\
		\vdots&\ddots&\vdots&\vdots&\ddots&\vdots&\vdots&\ddots&\vdots\\
		u_{N}^{\top} w^2_1\chi_{N,1} & \cdots &u_{N}^{\top} w^2_n\chi_{N,n}&w^2_1\chi_{N,1} &\cdots&w^2_n\chi_{N,n}&(u_N^{\top}w^1_1+b_1)\chi_{N,1} &\cdots&(u_N^{\top}w^1_n+b_n)\chi_{N,n}
	\end{bmatrix}.
\end{aligned} 
	\end{equation}
\end{figure*}
It is straightforward to observe that the columns in (\ref{full-dH}) are linearly dependent, since\footnote{Note that the following expression holds for $w_1^2\neq 0$. For $w_1^2=0$, the first column of $\mathrm{D}\mathscr{H}_{N}(w)$ in~(\ref{full-dH}) is zero, and hence the observability rank condition is trivially not satisfied.}

\begin{align*}
	\frac{1}{w^2_1}\begin{bmatrix}
		u_1^{\top} w^2_1\chi_{1,1}\\u_2^{\top} w^2_1\chi_{2,1}\\\vdots\\u_{N}^{\top} w^2_1\chi_{N,1}
	\end{bmatrix}w^1_1\!+\! \frac{b_1}{w^2_1}\begin{bmatrix}
		w^2_1\chi_{1,1}\\w^2_1\chi_{2,1}\\\vdots\\w^2_1\chi_{N,1}
	\end{bmatrix}\!\!=\!\!\begin{bmatrix}
		(u_1^{\top}w^1_1\!+\!b_1)\chi_{1,1}\\(u_2^{\top}w^1_1\!+\!b_1)\chi_{2,1}\\\vdots\\(u_N^{\top}w^1_1\!+\!b_1)\chi_{N,1}
	\end{bmatrix},
\end{align*}where $w^2_1$ and $b_1$ are scalars, and $w^1_1\in\mathbb{R}^m$. Consequently, the FNN in Fig.~\ref{fig:FNN2} with output weights cannot satisfy the local observability rank condition. Moreover, the linear dependence presented above implies that for multi-layer FNNs in general, the local observability rank condition fails to hold.

 This result implies that the use of state estimation-based methods for training FNNs with these characteristics is, in general, not well justified. However, such training has been successfully performed, e.g.,  in~\cite{bonassi2022towards,bemporad2022recurrent,singhal1988training}. Hence, novel techniques and alternative (weaker) observability conditions are required to guarantee the convergence of such algorithms for multi-layer FNNs, which will be studied in future work.
	
In Section~\ref{2-layer FNN}, the local observability analysis of a certain FNN configuration shows that the states of its reformulated dynamical system are locally observable under specific conditions. In the following section, we will exploit this property to design and analyze a state estimation-based training method.
 
\section{MHE-based FNN training}\label{section-MHE}
In this section, we first introduce the proposed MHE-based training method for regression FNNs, which tunes network parameters by estimating the optimal states of system~(\ref{2-eq-FNN problem}). Subsequently, we establish the convergence guarantees for the state estimation error.
\subsection{MHE problem formulation}
MHE is a state estimation technique that estimates the states of dynamical systems by solving an optimization problem over a moving window of the most recent input–output measurements. Here, we use it to train an FNN by estimating the optimal states of its reformulated dynamical system. According to the state propagation law in~(\ref{2-eq-FNN state}), the optimal state is static because the ideal weights of an FNN are constant. Consequently,  at time step $t\in\mathbb{Z}_{\geq 1}$, an optimal estimate can be obtained by solving the MHE problem using any batch of temporally independent input-output data and the previous optimal estimate $\hat{w}_{t-1}$. The current optimal state estimate $\hat{w}_t$ is obtained by solving the following optimization problem

\begin{subequations}\label{MHE problem}
	\begin{align}
	\min_{\bar{w}_{t}}\quad &V_t(\bar{w}_t)=\sum_{j\in \mathbb{Z}_t}\norm{\bar{y}_{j}-\tilde{y}_{j}}_2^2\label{MHE-c-cost}\\
	s.t. 			\quad	&\bar{y}_{j}=h(\bar{w}_{ t},u_j),~j\in \mathbb{Z}_{t},\label{MHE-b}\\
	& P_{\bar{o},t} (\bar{w}_t-\hat{w}_{t-1})=0,\label{MHE-c-unobservable}\\
	&\bar{w}_{t}\in \mathbb{W}_0\label{MHE-d}.
	\end{align}
\end{subequations}Here, $\hat{w}_{t-1}\in\mathbb{R}^{mn}$ denotes the optimal state estimate (obtained by solving (\ref{MHE problem})) at time  $t-1$. In case that the optimizer is not unique, an arbitrary minimizer can be selected. Furthermore,  $\tilde{y}_{j}$ is the output in the training dataset  corresponding to input $u_{j}$, and can be written as $\tilde{y}_{j}=y_j+\epsilon_{j}$, where $y_j$ is the output of the FNN under the \emph{ideal} state $w$ and input $u_j$, and $\epsilon_{j}$ is the bounded fitting error. Note that the following results hold independently of how exactly the \emph{ideal} weights are defined; these could, e.g., be hypothetical weights that minimize the fitting error over the whole domain of interest, or those that minimize the fitting error over the training data set.  The matrix $P_{\bar{o},t}$ denotes the projection   onto the unobservable subspace of system~(\ref{2-eq-FNN problem}), defined later in~(\ref{projection-matrix-un}). At $t=1$, let $\hat{w}_0=w_0$, where $w_0$ is the initial state. 

Let the training dataset be $\mathcal{D}_{tr}=\{(u_i,\tilde{y}_i)\}_{i=1}^{N}$, where $N\geq mn$. This lower bound corresponds to the minimum data requirement that allows local observability as discussed in Remark~\ref{rem-lower bound}.  Moreover, consider an integer $N_1 \leq N$ that represents the number of data points that is used to solve problem~(\ref{MHE problem}) at each time instant. Hence, the set $\mathbb{Z}_t\subset\mathbb{Z}_{[1,N]}$ in~(\ref{MHE-b}) contains $N_1$ integers of $\mathbb{Z}_{[1,N]}$.  Note that one could take all available training data into account and solve~(\ref{MHE problem}) once with $\mathbb{Z}_t$ replaced by $\mathbb{Z}_{[1,N]}$. However, this is generally not tractable since $N$ can be very large. Hence, at each time $t$, we use a mini-batch $\tilde{\mathcal{D}}_{t}=\{(u_j,\tilde{y}_j)\}, j\in\mathbb{Z}_t$, of size $N_1 < N$ for solving problem~(\ref{MHE problem}). This optimization setting is inspired by MHE, where finite estimation horizons are used. 
However, system (\ref{2-eq-FNN problem}) is static and its output depends only on the current input and state. Consequently, the order of the input-output sequence in $\tilde{\mathcal{D}}_{t}$ is not critical for the optimization at each step.  This differs from the classical MHE formulations, where sequences are temporally generated and present the dynamics flow of the system, cf. \cite{schiller2023lyapunov}. For the training dataset $\mathcal{D}_{tr}$ and  mini-batch dataset $\tilde{\mathcal{D}}_{t}$, we have the following assumption.

\begin{assumption}\label{assumption-coPE}
	The initial state $w_0$ in~(\ref{MHE problem}) is locally observable under the training dataset $\mathcal{D}_{tr}=\{(u_i,\tilde{y}_i)\}_{i=1}^{N}$ with a locally observable neighborhood $\mathbb{W}_0$. Moreover, $\mathcal{D}_{tr}$ is partitioned  into $k$ non-overlapping mini-batches with $N_1$ samples each. These mini-batches are used periodically to solve~(\ref{MHE problem}).  
\end{assumption}

A training epoch is completed after all mini-batches have been processed exactly once. The selection then returns to the first mini-batch whenever $t=\eta k+1$ for any integer $\eta\geq 0$, leading to a periodic schedule.  Moreover,  under Assumption~\ref{assumption-coPE}, the aggregate input sequence formed by concatenating all mini-batches is persistently exciting for the state $w_0$.  However, since $\tilde{\mathcal{D}}_{t}$ is only a subset of $\mathcal{D}_{tr}$, and not necessarily persistently exciting for $w_0$, satisfaction of the observability rank condition at state $w_0$ is not guaranteed for each optimization step. In the following subsection, we take this fact into account and establish the convergence of the MHE problem with~(\ref{MHE-c-unobservable}).

\subsection{Convergence analysis}
In the following, we  first characterize  the mapping between state and output differences using the mean value theorem, and derive the projection matrices onto the observable and unobservable subspaces. These projections are then used to analyze the MHE-based training framework.

\subsubsection{Relationship between states and outputs}
Based on the mean value theorem \cite{abraham2012manifolds}, for any two states $w_1 ,w_2\in \mathbb{W}_0$, where $\mathbb{W}_0$ is a locally observable neighborhood of $w_0$ defined in Definition~\ref{def-observable set}, the corresponding output difference $y_1-y_2$ under input $u_1$ satisfies
\begin{align}\label{eq-mean-1}
	y_1\!-\!y_2\!=\!\int_{0}^{1}\frac{\partial h}{\partial w} (\tau w_1+(1-\tau )w_2,u_1) \mathrm{d}\tau (w_1-w_2).
\end{align}Since $w_1 ,w_2\in \mathbb{W}_0$ and $\mathbb{W}_0$ is a convex set by Proposition~\ref{prop-convex}, the convex combination $(\tau w_1+(1-\tau)w_2)$ also belongs to $\mathbb{W}_0$ for any $\tau\in[0,1]$. As shown in the proof of  Theorem~\ref{th-observable set}, for any state $w^{\prime}$ in $\mathbb{W}_0$, the corresponding entries in $UW_0$ and $UW^{\prime}$ share the same signs under the same persistently exciting input sequence. Consequently, we have $\chi(UW_0)=\chi(UW^{\prime})$, which implies that all states in $\mathbb{W}_0$ share a common indicator matrix $T_{\chi}=\chi(UW_0)$ as defined in~(\ref{2-eq-dH-factorization}) (compare~\ref{def-Tchi}). Since $\frac{\partial h}{\partial w}(w^{\prime},u_1)$ can be expressed in terms of the input $u_1$ and a row of the indicator vector $T_{\chi}$ (see (\ref{2-eq-derivative})), its value is identical for all states in $\mathbb{W}_0$ under the same input. We therefore denote this common derivative  under $u_1$ by $\nabla h_1$. Combining the above results with~(\ref{eq-mean-1}), we have
\begin{align}
	y_1-y_2=\nabla h_1(w_1-w_2).
\end{align}

For the mini-batch   $\tilde{\mathcal{D}}_{t}=\{(u_j,\tilde{y}_j)\}, j\in\mathbb{Z}_t$ employed at time $t$, we denote the input sequence by $u_{[1,N_1],t}$. The output sequences generated by the FNN subject to $u_{[1,N_1],t}$ for states $w_1$ and $w_2$ are $\{h(w_1,u_{1,t}), \ldots, h(w_1,u_{N_1,t})\}$ and $\{h(w_2,u_{1,t}), \ldots, h(w_2,u_{N_1,t})\}$, respectively. The output differences can be expressed as
\begin{align}\label{eq-state output}
	\underbrace{\begin{bmatrix}
			h(w_1,u_{1,t})-h(w_2,u_{1,t})\\\vdots \\ h(w_1,u_{N_1,t})-h(w_2,u_{N_1,t})
	\end{bmatrix}}_{\Delta y_{[1,N_1],t}}
	=\underbrace{\begin{bmatrix}
			\nabla h_1\\ \vdots \\ \nabla h_{N_1}
	\end{bmatrix}}_{ H_t}(w_1-w_2),
\end{align}where the subscript $t$ denotes the time step when employing the data batch $\tilde{\mathcal{D}}_{t}$ for solving the optimization problem, and $H_t$ is the Jacobian matrix derived as shown in~(\ref{2-eq-H-derivative}). Consequently, when the input sequence $u_{[1,N_1],t}$ is persistently exciting for the initial state $w_0$, the matrix $H_t\in\mathbb{R}^{N_1\times mn}$ has full column rank. Furthermore, since all state pairs in the set $\mathbb{W}_0$ share the same $H_t$, their corresponding output sequences are mutually distinguishable. However, if the input sequence of the current mini-batch is not persistently exciting,  distinguishability is still maintained for the projection of these states onto the observable subspace. 

Note that the unobservable subspace of system~(\ref{2-eq-FNN problem}) under the data batch $\tilde{\mathcal{D}}_{t}$ is the null space of $H_t$, denoted by $\mathcal{N}(H_t)$. This follows because if  the state difference lies in the null space of $H_t$, the output difference is zero, rendering  the states indistinguishable. The orthogonal projection matrix onto the unobservable subspace at time $t$ is given by
\begin{align}
 P_{\bar{o},t}=I-H_t^{\dagger}H_t^{\top},\label{projection-matrix-un}
\end{align}which is used in the constraint~(\ref{MHE-c-unobservable}). The observable subspace is the row space of $H_t$, denoted by $\mathcal{R}(H_t)$, with the orthogonal projection matrix given by
\begin{align}
	 P_{o,t}=H_t^{\dagger}H_t^{\top}.\label{projection-matrix-o}
\end{align} 

 Using $P_{\bar{o},t}$ and $P_{o,t}$, we decompose the state at time step $t$ into unobservable and observable components. For a mini-batch $\tilde{\mathcal{D}}_{t}$ that lacks persistent excitation for the states in $\mathbb{W}_0$, the current data batch provides no effective information for the estimation of the unobservable components.  Constraint~(\ref{MHE-c-unobservable}) is designed to fix the projection onto the unobservable subspace. Consequently, at each optimization step, only the observable components are updated, which directly influence the value of the cost function. 
\subsubsection{Convergence of MHE}
The proposed MHE-based training scheme~(\ref{MHE problem}) estimates the ideal FNN weights by solving an optimization problem at each time step $t$ using the data batch $\tilde{\mathcal{D}}_{t}$. Naturally, convergence to the ideal weights cannot be expected from arbitrary initial conditions, as observability is only guaranteed locally. Nevertheless, local convergence can be shown if the initial weights are close enough to the ideal weights, as formalized in the following assumption.
	\begin{assumption}\label{assum-true state}
		The ideal state $w$ of the reformulated system~(\ref{2-eq-FNN problem}) satisfies $w\in\mathbb{W}_0$.
	\end{assumption}

\begin{theorem}\label{th5}
	Let Assumption~\ref{assum-true state} hold. Then, for all $t\in \mathbb{Z}_{\geq 1}$, the state estimate $\hat{w}_t$ that minimizes~(\ref{MHE problem}) satisfies
	\begin{align}
		\norm{P_{o,t}(\hat{w}_t-w)}_2\leq \max_{i\in\mathbb{Z}_{t}}2{\abs{\epsilon_{i}}}\sqrt{N_1}/\sigma_t ,\label{constrian-ob}
	\end{align}where $\sigma_t$ is the minimal nonzero singular value of $H_t$.
\end{theorem}	
\begin{proof}
Since $ \mathcal{R}(H_t)\oplus \mathcal{N}(H_t)=\mathbb{R}^n$, $\hat{w}_t$ can be expressed as $\hat{w}_t=P_{\bar{o},t} \hat{w}_{t}+P_{o,t} \hat{w}_t$. By~(\ref{eq-state output}), the optimal cost function~(\ref{MHE-c-cost}) can be represented as $V_t^*(\hat{w}_t)=\sum_{j\in \mathbb{Z}_t}\norm{\hat{y}_{j}-\tilde{y}_{j}}_2^2=\norm{H_t(\hat{w}_t-w)-\varepsilon_t}_2^2=\norm{H_tP_{o,t}(\hat{w}_t-w)-\varepsilon_t}_2^2$, where $\varepsilon_t\in\mathbb{R}^{N_1}$ is a vector whose entries are the fitting errors $\{\epsilon_{j}\},j\in\mathbb{Z}_t$. The last equality holds since the unobservable component $P_{\bar{o},t} (\hat{w}_{t}-w)$ lies in the null space of $H_t$. By Assumption~\ref{assum-true state}, we have $w\in\mathbb{W}_0$. Together with the convexity of $\mathbb{W}_0$, it follows that $w^{\prime}=P_{\bar{o},t} \hat{w}_{t-1}+P_{o,t} w$ satisfies all constraints in~(\ref{MHE problem}) and is a feasible solution. Hence, $ V_t^*(\hat{w}_t)\leq V_t(w^{\prime})$ by optimality.  Taking square roots of both sides, we obtain $
	\norm{H_tP_{o,t}(\hat{w}_t-w)-\varepsilon_t}_2\leq\norm{\varepsilon_t}_2$, where the cost of $w^{\prime}$ is $V_t(w^{\prime})=\norm{\varepsilon_t}_2^2$. Using the fact that $
	\norm{H_tP_{o,t}(\hat{w}_t-w)}_2-\norm{\varepsilon_t}_2\leq \norm{H_tP_{o,t}(\hat{w}_t-w)-\varepsilon_t}_2$,  we have
	\begin{align}\label{HPbound1}
		 \norm{H_tP_{o,t}(\hat{w}_t-w)}_2\leq 2\norm{\varepsilon_t}_2.
	\end{align}  Suppose $\mathrm{rank}(H_t)=r$. Since $P_{o,t}(\hat{w}_t-w)\in\mathcal{R}(H_t)$, there exists a vector $z_t\in\mathbb{R}^r$, such that $P_{o,t}(\hat{w}_t-w)=Q_tz_t$, where $Q_t\in \mathbb{R}^{n \times r}$ is the matrix of right singular vectors of $H_t$ from its singular value decomposition $H_t=R_t\Sigma_t Q_t^{\top}$, and $\Sigma_t$ contains the $r$ nonzero singular values of $H_t$. Since $R_t$ and $Q_t$ are semi-unitary matrices, we have $\norm{H_tP_{o,t}(\hat{w}_t-w)}_2=\norm{R_t\Sigma_tQ_t^{\top}Q_tz_t}_2=\norm{\Sigma_tz_t}_2$. Combining this result with~(\ref{HPbound1}), we have
	\begin{align}
		\sigma_t\norm{z_t}_2\leq\norm{\Sigma_tz_t}_2\leq2\norm{\varepsilon_t}_2,
	\end{align}where $\sigma_t$ is the minimal nonzero singular value of $H_t$. Since $Q_t$ is a semi-unitary matrix, it follows that $\norm{P_{o,t}(\hat{w}_t-w)}_2=\norm{z_t}_2$. Using the fact that $\norm{\varepsilon_t}_2\leq \max_{i\in\mathbb{Z}_{t}}\abs{\epsilon_{i}}\sqrt{N_1}$, we obtain~(\ref{constrian-ob}) and complete the proof.
\end{proof}

 Since the input sequence in the training dataset $\mathcal{D}_{tr}$ is persistently exciting for states in $\mathbb{W}_0$, the Jacobian matrix $\mathrm{D}\mathscr{H}_N(w)$ for the ideal state (and every other state in $\mathbb{W}_0$) has full column rank. This implies that the corresponding unobservable subspace under $\mathcal{D}_{tr}$ is empty, with its projection matrix denoted by $P_{\bar{o}}=\boldsymbol{0}$.  For each $t\in \mathbb{Z}_{\geq 1}$ and the mini-batch $\tilde{\mathcal{D}}_{t}$, there exists a corresponding $H_t$ and a pair of projection matrices $P_{\bar{o},t}$ and $P_{o,t}$.  From Assumption~\ref{assumption-coPE}, each $H_t$ is a row submatrix of $\mathrm{D}\mathscr{H}_N(w)$, and consequently for $t\geq k$, we have $\mathcal{R}(H_1)+\mathcal{R}(H_2)+\cdots+\mathcal{R}(H_{t})=\mathcal{R}(\mathrm{D}\mathscr{H}_N(w))$. Since $ \mathcal{R}(H_t)\oplus \mathcal{N}(H_t)=\mathbb{R}^n$, applying De Morgan's law for orthogonal complements yields $\mathcal{N}(\mathrm{D}\mathscr{H}_N(w))=\bigcap_{\tau=1}^{t} \mathcal{N}(H_{\tau})=0$.

We next show that consecutive projections of a state onto the null spaces of $H_t$ converge in norm to the projection onto the intersection of these subspaces. This result follows directly  from the following lemma about the alternating projection theory. 

\begin{lemma}[\!\!{\cite[Th. 1]{sakai1995strong}}]\label{lemma-alternating projection}
	Let $\mathcal{H}_{j}$, $j\in\{1,\cdots,J\}$ be a finite number of closed linear subspaces of a Hilbert space, and let $P_j$ be the corresponding orthogonal projection matrix onto $\mathcal{H}_{j}$.  Let $s=(j_i)_{i\geq1}, j_i\in \{1,\cdots,J\}$, be an infinite integer sequence. If $s$ is quasi-periodic, then for an element $x_0$, the consecutive projection $\Pi_{i=1}^{\infty}P_{j_i}$ converges in norm to the orthogonal projection of $x_0$ onto the intersection of these subspaces $\bigcap_{j=1}^J\mathcal{H}_{j}$.
\end{lemma}

Recalling that the training dataset is partitioned into a finite number of mini-batches, each processed once per epoch, the input-output pairs in $\tilde{\mathcal{D}}_t$ are used periodically as $t$ approaches infinity. Define $\mu\coloneqq\max_{\tau\in\mathbb{Z}_{[1,k]}}2\sqrt{N_1}/\sigma_{\tau}$, where $k$ is the number of distinct mini-batches and $\sigma_{\tau}$ is the minimal nonzero singular value of $H_{\tau}$. For any $t\in\mathbb{Z}_{\geq k}$, it holds that $\mu=\max_{\tau\in\mathbb{Z}_{[1,t]}}2\sqrt{N_1}/\sigma_{\tau}$. 

As discussed above, since the projection onto each unobservable subspace is employed periodically, Lemma~\ref{lemma-alternating projection} implies that $\lim\limits_{t\rightarrow\infty}\norm{P_{\bar{o}}w_0-\Pi_{\tau=1}^{t}P_{\bar{o},\tau}w_0}=0$. Based on this result, the following theorem shows that under the proposed MHE-based training method and the mini-batch dataset generation scheme,  the state estimation error is guaranteed to converge to a neighborhood of zero, the size of which vanishes as  the bound of measurement noise tends to zero.

\begin{theorem}\label{th-MHE}
	Consider the FNN defined in (\ref{2-eq-FNN problem}) with weight matrix $W$, a persistently exciting input sequence $u_{[1,N]}$, and training dataset $\mathcal{D}_{tr}$. Let Assumptions~\ref{assumption-coPE} and~\ref{assum-true state} hold. Then, the weight estimation error under the MHE-based method~(\ref{MHE problem}) satisfies:
	\begin{align}\label{th-error bound}
		\limsup\limits_{t\rightarrow\infty}\norm{w-\hat{w}_{t}}_2\leq\frac{k\mu \zeta}{1-\rho}\max_{i\in\mathbb{Z}_{[1,N]}}\abs{\epsilon_{i}} ,
	\end{align}where  $\zeta\geq 1$ and $\rho\in[0,1)$ are some constants.
\end{theorem}
\begin{proof} Let $\alpha_t\coloneqq P_{o,t}(w-\hat{w}_t)$. Together with constraint~(\ref{MHE-c-unobservable}), we obtain 
\begin{align}
	w-\hat{w}_{t}&=(P_{o,t}+P_{\bar{o},t})(w-\hat{w}_{t})\\
	&=P_{\bar{o},t}(w-\hat{w}_{t})+\alpha_{t}\\
	&=P_{\bar{o},t}w-P_{\bar{o},t}(P_{o,t-1}+P_{\bar{o},t-1})\hat{w}_{t-1}+\alpha_{t}.
\end{align} By employing~(\ref{MHE-c-unobservable}) recursively, we have $w-\hat{w}_{t}=P_{\bar{o},t}w-P_{\bar{o},t}(P_{o,t-1}w-\alpha_{t-1}+P_{\bar{o},t-1}(P_{o,t-2}w-\alpha_{t-2}+\cdots P_{\bar{o},2}(P_{o,1}w-\alpha_1+P_{\bar{o},1}\hat{w}_{0})\cdots))+\alpha_{t}$. Then, it follows that
\begin{align}\label{th5-eq ee}
	 	&w-\hat{w}_{t}\notag\\
	 	&=P_{\bar{o},t}w-P_{\bar{o},t}(P_{o,t-1}w+\cdots P_{\bar{o},2}(P_{o,1}w+P_{\bar{o},1}w)\cdots)\notag\\
	 	&\quad+\alpha_{t}+P_{\bar{o},t}\alpha_{t-1}+\cdots+P_{\bar{o},t}\cdots P_{\bar{o},2}\alpha_{1}\notag\\
	 	&\quad+P_{\bar{o},t}\cdots P_{\bar{o},1}(w-\hat{w}_{0}).
\end{align}
Since $P_{o,i}w+P_{\bar{o},i}w=w$, for all $i\in \mathbb{Z}_{[1,t]}$, the first row in~(\ref{th5-eq ee}) reduces to $P_{\bar{o},t}w-P_{\bar{o},t}w=0$. 
As mentioned above, $P_{\bar{o}}=\boldsymbol{0}$, since the weight matrix $W$ is locally observable under the training dataset $\mathcal{D}_{tr}$ according to Assumption~\ref{assum-true state}.  
Under Assumption~\ref{assumption-coPE}, Lemma~\ref{lemma-alternating projection} ensures that the product of alternating projections converges to the projection onto the intersection of subspaces. Consequently, 
\begin{align}\label{alternating P}
	\lim\limits_{t\rightarrow\infty}\norm{ P_{\bar{o},t}\cdots P_{\bar{o},1}(w-\hat{w}_{0})}_2=\norm{P_{\bar{o}}(w-\hat{w}_{0})}_2=0.
\end{align}

We now turn to the second row of~(\ref{th5-eq ee}) and consider $t=\eta k+c$ with $\eta\geq0$ and $c\in\mathbb{Z}_{[1,k-1]}$. As the MHE problem~(\ref{MHE problem}) is solved periodically over $k$ distinct mini-batches, the projection matrices satisfy $P_{\bar{o},i}=P_{\bar{o},i+kj}$ for $i\in\mathbb{Z}_{[1,k]}, j\in\mathbb{Z}_{[0,\eta]}$.  Define $Q\coloneqq \Pi_{i=1}^{k}P_{\bar{o},i}$. By (\ref{alternating P}), it follows that $\lim\limits_{\eta\rightarrow\infty} \Pi_{\tau=1}^{\eta k} P_{\bar{o},\tau}=\lim\limits_{\eta\rightarrow\infty}Q^{\eta}=0$. Therefore, the spectral radius $\rho$ of $Q$ lies in $\rho\in[0,1)$. Consider the case where every single data batch $\tilde{\mathcal{D}}_t$ is not persistently exciting for $w_0$, resulting in $P_{\bar{o},t}\neq \boldsymbol{0}$ for any $t\in\mathbb{Z}_{\geq 1}$ and $\rho\in(0,1)$.  By Gelfand's formula \cite{horn2012matrix}, there exists a constant $\zeta\geq 1$ such that $\norm{Q^{\ell}}_2\leq \zeta \rho^{\ell}$, for all $\ell\in\mathbb{Z}_{\geq0}$. Based on the results above, we obtain
\begin{align}
	&I+P_{\bar{o},t}+\cdots+P_{\bar{o},t}\cdots P_{\bar{o},2}\notag\\
	&=I+P_{\bar{o},\eta k+c}+\cdots+\Pi_{i=2+c}^{k+c}P_{\bar{o},(\eta-1) k+i}\notag\\
	&\quad+Q+QP_{\bar{o},(\eta-1) k+c}+\cdots+Q\Pi_{i=2+c}^{k+c}P_{\bar{o},(\eta-2) k+i}\notag\\
	&\quad+\cdots\notag\\
	&\quad+Q^{\eta-1}+Q^{\eta-1}P_{\bar{o}, k+c}+\cdots+Q^{\eta-1}\Pi_{i=2+c}^{k+c}P_{\bar{o},i}\notag\\
	&\quad+Q^{\eta}+Q^{\eta}P_{\bar{o}, c}+\cdots+Q^{\eta}\Pi_{i=2}^{c}P_{\bar{o},i},\label{bound2}
\end{align}where each term on the right-hand side corresponds to its equivalent representation on the left-hand side. Taking the  norm of  each term in~(\ref{bound2}) and applying Gelfand's formula yields
\begin{align}\label{residual1}
	&\norm{I}_2+\norm{P_{\bar{o},t}}_2+\cdots+\norm{P_{\bar{o},t}\cdots P_{\bar{o},2}}_2\notag\\
	&\leq k \zeta\frac{1-\rho^{\eta}}{1-\rho}+c\zeta\rho^{\eta},
\end{align}
since $\norm{\Pi_{i=\ell}^{k}P_{\bar{o}, i+jk}}_2\leq 1$ for $\ell\in\mathbb{Z}_{[2,k]}, j\in\mathbb{Z}_{[0,\eta]}$. From Theorem~\ref{th5} and the definition of $\mu$, we have $\norm{\alpha_t}_2\leq\mu\max_{i\in\mathbb{Z}_{[1,N]}}\abs{\epsilon_{i}}$. Combining this bound  with~(\ref{th5-eq ee}), (\ref{alternating P}) and~(\ref{residual1}) yields
\begin{align}\label{th6-bound}
	\limsup\limits_{t\rightarrow\infty}\norm{w-\hat{w}_{t}}_2\leq\frac{k\mu \zeta}{1-\rho}\max_{i\in\mathbb{Z}_{[1,N]}}\abs{\epsilon_{i}}.
\end{align}Now, consider the case where some data batch $\tilde{\mathcal{D}}_t$ is persistently exciting for $w_0$, which implies the corresponding $P_{\bar{o},t}= \boldsymbol{0}$. Hence, $Q=\boldsymbol{0}$ and~(\ref{residual1}) holds with the right-hand side replaced by $k$. Therefore, (\ref{th6-bound}) also holds  for $\rho=0$, which completes the proof.\end{proof}

\begin{remark}
The convergence result in Theorem~\ref{th-MHE}  provides a bound on the difference between the ideal weights $w$ and the weights $\hat{w}_t$ estimated via~(\ref{MHE problem}) in terms of the fitting error $\epsilon$ associated with the ideal weights. Theorem~\ref{th-MHE} is developed under the periodic mini-batch employment described in Assumption~\ref{assumption-coPE}. However, a similar analysis can be extended to quasi-periodic mini-batch employment, since only finitely many distinct products of projection matrices (compare $Q$ in the proof of Theorem~\ref{th-MHE}) arise, which admits a uniform constant $\rho\in[0,1)$.  
\end{remark}

\section{Numerical results}\label{section-simulation}
In this section, we first validate our theoretical results using synthetic data, and then demonstrate the effectiveness of the MHE-based method~(\ref{MHE problem}) on a regression benchmark dataset from UCI \cite{asuncion2007uci}. All experiments were performed on a laptop  running Windows 11 Pro (64-bit) with a 13th Gen Intel Core i7-13580HX processor and 16 GB of RAM, using MATLAB R2024a. The MHE optimization problem was solved using CasADi \cite{andersson2019casadi}.
 
\subsection{A synthetic data case}
In this experiment, we consider an FNN with $2$ inputs, $10$ hidden nodes, and $1$ output. Every hidden node uses a ReLU activation function, and the weights from hidden layer to the output are fixed to $1$. We generate a teacher FNN with random, locally observable weights, which serve as the ideal weights (i.e., the ideal state of the reformulated dynamical system~(\ref{2-eq-FNN problem})). A student FNN with identical architecture is then trained to learn the weights of the teacher FNN. 

To demonstrate the convergence of the state estimation error, we generate a persistently exciting input sequence for the ideal state (teacher weights) and collect the corresponding outputs with bounded measurement noise ($\max_{i\in\mathbb{Z}_{[1,N]}}\abs{\epsilon_{i}}=1\times 10^{-4}$), forming a training dataset $\mathcal{D}_{tr}$ of $90$ input-output pairs. This dataset is partitioned into $5$ mini-batch datasets, each containing $N_1=18$ input-output pairs. We select an initial locally observable state $w_0$ sufficiently close  to the ideal state so that the ideal state lies in the locally observable neighborhood $\mathbb{W}_0$ of $w_0$, constructed according to Theorem~\ref{th-observable set}. Consequently, Assumptions~\ref{assumption-coPE} and~\ref{assum-true state}  hold. 

We implement the MHE-based training method~(\ref{MHE problem}) on the dataset $\mathcal{D}_{tr}$ to estimate the ideal state $w$. At each step, the state estimate is updated by solving the optimization problem for a single mini-batch. Processing all mini-batches sequentially constitutes one training  epoch. The experiment runs for $2000$ epochs, during which each mini-batch is used periodically. This schedule ensures that, in the limit of infinite epochs, every mini-batch appears periodically, guaranteeing convergence of the state estimate to a neighborhood of the ideal state by Theorem~\ref{th-MHE}. A test dataset consisting of $90$ input-output pairs is generated randomly.  Fig.~\ref{fig:loss} and Fig.~\ref{fig:ee} compare the performance of the proposed method against a standard mini-batch GD method. For a fair comparison, the mini-batch GD uses the same batch size $N_1$ as the MHE-based method,  with a learning rate of $0.1$.  Although the average time consumption per epoch is higher for the MHE-based method, its state estimation error converges faster, reaching $5\times 10^{-3}$ within one epoch, as shown in Table~\ref{table time consumption}. In summary, the proposed MHE-based FNN training method achieves strong convergence in weight estimation error, training loss, and test loss, most importantly, with theoretical guarantees on the weight convergence. 
\begin{table}[!t]
		\caption{Comparison of time consumption between MHE-based training method and the mini-batch GD method.}
	\centering
	\begin{tabular*}{\columnwidth}{@{\extracolsep{\fill}}lccc@{}}
		\toprule[1pt]  
		Method & Time to estimation err.=5$\times 10^{-3}$ & Avg. time/epoch \\
		\midrule        
		MHE-based training & 1.5 s & 1.5 s \\
		mini-batch GD & 3.3 s & 1.6 ms \\
		\bottomrule[1pt] 
	\end{tabular*}
	\label{table time consumption}
\end{table}
\begin{figure}[!t]
	\centering
	\includegraphics[width=3.2in]{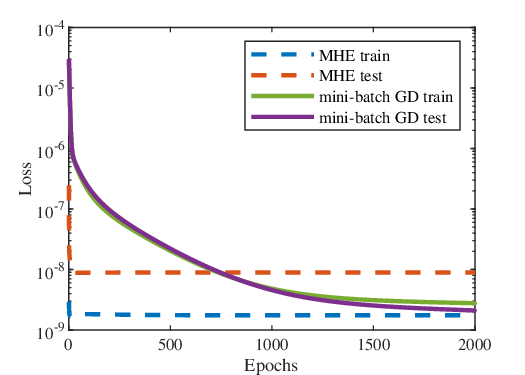}
	\caption{Comparison of loss between MHE-based training method and the mini-batch GD method.}
	\label{fig:loss}
\end{figure}
\begin{figure}[!t]
	\centering
	\includegraphics[width=3.2in]{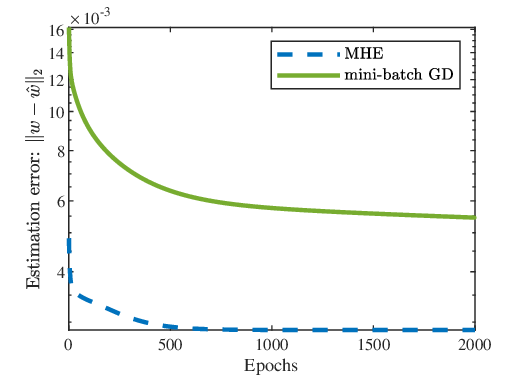}
	\caption{The error between the ideal weights and the estimate weights.}
	\label{fig:ee}
\end{figure}
\subsection{A UCI dataset case}
To further demonstrate the potential of the MHE-based training method, we evaluate it on the UCI Wine benchmark dataset \cite{asuncion2007uci}. We use a fully connected two-layer FNN with 32 hidden nodes for both experiments. ReLU activation functions are used in the hidden layers, while the output layer employs a linear activation function.  As established in Section~\ref{section-multi FNN}, local observability is not guaranteed for such FNN configurations. However, we show that with minor modifications, the proposed MHE-based training scheme achieves competitive performance compared with existing methods. Datasets are divided into training (90\%) and test (10\%) subsets.  We therefore adopt the following modified cost function for our proposed method~(\ref{MHE-c-cost}): 
\begin{align}
	\min_{\bar{w}_{t}} V_t(\bar{w}_t)=\sum_{j\in \mathbb{Z}_t}10^{-4}\norm{\bar{y}_{j}-\tilde{y}_{j}}_2^2+\norm{\bar{w}_t-\hat{w}_{t-1}}_2^2
\end{align} with a horizon length of $32$, and remove the observability-related constraints~(\ref{MHE-c-unobservable}) and~(\ref{MHE-d}). We compare against standard mini-batch training using the Adam optimizer \cite{kingma2014adam} with a learning rate of $0.001$, a batch size of $32$,  and $150$ epochs.  Each experiment is run $30$ times, and the mean test error and standard deviation are presented in Table~\ref{table benchmark}. We also include recent results from~\cite{zhang2018noisy} on the same dataset using noisy
natural gradient descent with a matrix-variate Gaussian posterior (NNG-MVG). The results show that  the proposed MHE-based training method achieves a smaller RMSE. This simulation motivates us to further explore convergence guarantees for MHE-based training of general FNNs under weaker observability conditions.
\begin{table}[!t]
	\caption{Comparison of test RMSE on UCI Wine dataset.}
	\centering
	\begin{tabular*}{\columnwidth}{@{\extracolsep{\fill}}ccc@{}}
			\toprule[1pt]  
			 Adam optimizer & NNG-MVG&MHE-based training \\
			\midrule        
			  0.639$\pm$0.010 & 0.637$\pm$0.011 &\textbf{0.628$\pm$0.015}\\
			\bottomrule[1pt] 
		\end{tabular*}
		\label{table benchmark}
	\end{table}

\section{Conclusion}\label{section-conclustion}
In this paper, we presented a theoretical convergence analysis for FNN training using a novel MHE-based approach. The analysis employs tools from systems and control theory by first reformulating FNNs as dynamical systems, with the weights as states.  We investigated local observability of general FNN architectures with ReLU activations to establish convergence and stability guarantees for the MHE-based training approach. For a specific class of two-layer FNNs, we derived sufficient conditions under which local observability holds, enabling rigorous MHE convergence analysis. Furthermore, we developed a PE input design method for a locally observable state, and constructed its locally observable neighborhood, within which all states  are mutually distinguishable under the same PE inputs.

The proposed MHE-based FNN training method updates state estimates using mini-batch training datasets. These batches are chosen to be collectively PE, rendering the feasible candidate states locally observable, even when individual batches are not PE. This is achieved by fixing the state projection onto the unobservable subspace while  updating only the projection onto the observable subspace. The method proposed here is endowed with formal convergence guarantees under local observability conditions. The effectiveness of the proposed method was illustrated via a synthetic example and a benchmark example. 

An important challenge for  future research is to derive convergence guarantees for  FNN weights even when local observability cannot be ensured.

\bibliographystyle{IEEEtran}
\bibliography{reference}

\begin{thebibliography}{10}
\providecommand{\url}[1]{#1}
\csname url@samestyle\endcsname
\providecommand{\newblock}{\relax}
\providecommand{\bibinfo}[2]{#2}
\providecommand{\BIBentrySTDinterwordspacing}{\spaceskip=0pt\relax}
\providecommand{\BIBentryALTinterwordstretchfactor}{4}
\providecommand{\BIBentryALTinterwordspacing}{\spaceskip=\fontdimen2\font plus
\BIBentryALTinterwordstretchfactor\fontdimen3\font minus
  \fontdimen4\font\relax}
\providecommand{\BIBforeignlanguage}[2]{{%
\expandafter\ifx\csname l@#1\endcsname\relax
\typeout{** WARNING: IEEEtran.bst: No hyphenation pattern has been}%
\typeout{** loaded for the language `#1'. Using the pattern for}%
\typeout{** the default language instead.}%
\else
\language=\csname l@#1\endcsname
\fi
#2}}
\providecommand{\BIBdecl}{\relax}
\BIBdecl

\bibitem{bonassi2022recurrent}
F.~Bonassi, M.~Farina, J.~Xie, and R.~Scattolini, ``On recurrent neural
  networks for learning-based control: recent results and ideas for future
  developments,'' \emph{Journal of Process Control}, vol. 114, pp. 92--104,
  2022.

\bibitem{hunt1992neural}
K.~J. Hunt, D.~Sbarbaro, R.~{\.Z}bikowski, and P.~J. Gawthrop, ``Neural
  networks for control systems—a survey,'' \emph{Automatica}, vol.~28, no.~6,
  pp. 1083--1112, 1992.

\bibitem{krizhevsky2017imagenet}
A.~Krizhevsky, I.~Sutskever, and G.~E. Hinton, ``Imagenet classification with
  deep convolutional neural networks,'' \emph{Communications of the ACM},
  vol.~60, no.~6, pp. 84--90, 2017.

\bibitem{vaswani2017attention}
A.~Vaswani, N.~Shazeer, N.~Parmar, J.~Uszkoreit, L.~Jones, A.~N. Gomez,
  {\L}.~Kaiser, and I.~Polosukhin, ``Attention is all you need,''
  \emph{Advances in neural information processing systems}, vol.~30, 2017.

\bibitem{11159581}
H.~Hose, J.~Köhler, M.~N. Zeilinger, and S.~Trimpe, ``Approximate nonlinear
  model predictive control with safety-augmented neural networks,'' \emph{IEEE
  Transactions on Control Systems Technology}, vol.~33, no.~6, pp. 2490--2497,
  2025.

\bibitem{9993046}
D.~Tabas and B.~Zhang, ``Safe and efficient model predictive control using
  neural networks: An interior point approach,'' in \emph{2022 IEEE 61st
  Conference on Decision and Control (CDC)}, 2022, pp. 1142--1147.

\bibitem{li2022using}
Y.~Li, K.~Hua, and Y.~Cao, ``Using stochastic programming to train neural
  network approximation of nonlinear {MPC} laws,'' \emph{Automatica}, vol. 146,
  p. 110665, 2022.

\bibitem{schmidhuber2015deep}
J.~Schmidhuber, ``Deep learning in neural networks: an overview,'' \emph{Neural
  Networks}, vol.~61, pp. 85--117, 2015.

\bibitem{li2017convergence}
Y.~Li and Y.~Yuan, ``Convergence analysis of two-layer neural networks with
  {R}e{LU} activation,'' in \emph{Proceedings of the 31st International
  Conference on Neural Information Processing Systems}, 2017, pp. 597--607.

\bibitem{singhal1988training}
S.~Singhal and L.~Wu, ``Training multilayer perceptrons with the extended
  {K}alman algorithm,'' in \emph{Proceedings of the 2nd International
  Conference on Neural Information Processing Systems}, 1988, p. 133–140.

\bibitem{bemporad2022recurrent}
A.~Bemporad, ``Recurrent neural network training with convex loss and
  regularization functions by extended {K}alman filtering,'' \emph{IEEE
  Transactions on Automatic Control}, vol.~68, no.~9, pp. 5661--5668, 2022.

\bibitem{schiller2023lyapunov}
J.~D. Schiller, S.~Muntwiler, J.~K{\"o}hler, M.~N. Zeilinger, and M.~A.
  M{\"u}ller, ``A {L}yapunov function for robust stability of moving horizon
  estimation,'' \emph{IEEE Transactions on Automatic Control}, vol.~68, no.~12,
  pp. 7466--7481, 2023.

\bibitem{bonassi2022towards}
F.~Bonassi, J.~Xie, M.~Farina, and R.~Scattolini, ``Towards lifelong learning
  of recurrent neural networks for control design,'' in \emph{2022 European
  control conference (ECC)}.\hskip 1em plus 0.5em minus 0.4em\relax IEEE, 2022,
  pp. 2018--2023.

\bibitem{kleinberg2018alternative}
B.~Kleinberg, Y.~Li, and Y.~Yuan, ``An alternative view: when does {SGD} escape
  local minima?'' in \emph{International Conference on Machine Learning}.\hskip
  1em plus 0.5em minus 0.4em\relax PMLR, 2018, pp. 2698--2707.

\bibitem{albertini1993uniqueness}
F.~Albertini, E.~D. Sontag, and V.~Maillot, ``Uniqueness of weights for neural
  networks,'' in \emph{Artificial Neural Networks for Speech and Vision}, 1993.

\bibitem{bona2023parameter}
J.~Bona-Pellissier, F.~Bachoc, and F.~Malgouyres, ``Parameter identifiability
  of a deep feedforward {ReLU} neural network,'' \emph{Machine Learning}, vol.
  112, no.~11, pp. 4431--4493, 2023.

\bibitem{hermann1977nonlinear}
R.~Hermann and A.~Krener, ``Nonlinear controllability and observability,''
  \emph{IEEE Transactions on Automatic Control}, vol.~22, no.~5, pp. 728--740,
  1977.

\bibitem{nijmeijer1982observability}
H.~Nijmeijer, ``Observability of autonomous discrete time non-linear systems: a
  geometric approach,'' \emph{International Journal of Control}, vol.~36,
  no.~5, pp. 867--874, 1982.

\bibitem{albertini1996remarks}
F.~Albertini and D.~D’Alessandro, ``Remarks on the observability of nonlinear
  discrete time systems,'' in \emph{System Modelling and Optimization:
  Proceedings of the Seventeenth IFIP TC7 Conference on System Modelling and
  Optimization, 1995}.\hskip 1em plus 0.5em minus 0.4em\relax Springer, 1996,
  pp. 155--162.

\bibitem{sontag1984concept}
E.~D. Sontag, ``A concept of local observability,'' \emph{Systems \& Control
  Letters}, vol.~5, no.~1, pp. 41--47, 1984.

\bibitem{5400067}
A.~J. Krener and K.~Ide, ``Measures of unobservability,'' in \emph{Proceedings
  of the 48h IEEE Conference on Decision and Control (CDC) held jointly with
  2009 28th Chinese Control Conference}, 2009, pp. 6401--6406.

\bibitem{7403218}
N.~D. Powel and K.~A. Morgansen, ``Empirical observability {G}ramian rank
  condition for weak observability of nonlinear systems with control,'' in
  \emph{2015 54th IEEE Conference on Decision and Control (CDC)}, 2015, pp.
  6342--6348.

\bibitem{vanelli2025local}
M.~Vanelli and J.~M. Hendrickx, ``Local identifiability of fully-connected
  feed-forward networks with nonlinear node dynamics,'' in \emph{2025 European
  Control Conference (ECC)}.\hskip 1em plus 0.5em minus 0.4em\relax IEEE, 2025,
  pp. 825--830.

\bibitem{yang2025local}
Y.~Yang, V.~G. Lopez, and M.~A. M{\"u}ller, ``Local observability of a class of
  feedforward neural networks,'' in \emph{2025 IEEE 64th Conference on Decision
  and Control (CDC)}.\hskip 1em plus 0.5em minus 0.4em\relax IEEE, 2025, pp.
  90--95.

\bibitem{aichmayr2024sagemath}
M.~S. Aichmayr, S.~M{\"u}ller, and G.~Regensburger, ``A sagemath package for
  elementary and sign vectors with applications to chemical reaction
  networks,'' in \emph{International Congress on Mathematical Software}.\hskip
  1em plus 0.5em minus 0.4em\relax Springer, 2024, pp. 155--164.

\bibitem{rockafellar1969elementary}
R.~T. Rockafellar, ``The elementary vectors of a subspace of ${R}^n$,'' in
  \emph{Combinatorial Mathematics and Its Applications}.\hskip 1em plus 0.5em
  minus 0.4em\relax University of North Carolina Press, 1969, pp. 104--127.

\bibitem{brualdi1995sparse}
R.~A. Brualdi, S.~Friedland, and A.~Pothen, ``The sparse basis problem and
  multilinear algebra,'' \emph{SIAM Journal on Matrix Analysis and
  Applications}, vol.~16, no.~1, pp. 1--20, 1995.

\bibitem{cybenko1989approximation}
G.~Cybenko, ``Approximation by superpositions of a sigmoidal function,''
  \emph{Mathematics of Control, Signals and Systems}, vol.~2, no.~4, pp.
  303--314, 1989.

\bibitem{greville1960some}
T.~Greville, ``Some applications of the pseudoinverse of a matrix,'' \emph{SIAM
  Review}, vol.~2, no.~1, pp. 15--22, 1960.

\bibitem{abraham2012manifolds}
R.~Abraham, J.~E. Marsden, and T.~Ratiu, \emph{Manifolds, tensor analysis, and
  applications}.\hskip 1em plus 0.5em minus 0.4em\relax Springer Science \&
  Business Media, 2012, vol.~75.

\bibitem{sakai1995strong}
M.~Sakai, ``Strong convergence of infinite products of orthogonal projections
  in {H}ilbert space,'' \emph{Applicable Analysis}, vol.~59, no. 1-4, pp.
  109--120, 1995.

\bibitem{horn2012matrix}
R.~A. Horn and C.~R. Johnson, \emph{Matrix analysis}.\hskip 1em plus 0.5em
  minus 0.4em\relax Cambridge university press, 2012.

\bibitem{asuncion2007uci}
A.~Asuncion, D.~Newman \emph{et~al.}, ``{UCI} machine learning repository,''
  2007.

\bibitem{andersson2019casadi}
J.~A. Andersson, J.~Gillis, G.~Horn, J.~B. Rawlings, and M.~Diehl, ``Cas{AD}i:
  a software framework for nonlinear optimization and optimal control,''
  \emph{Mathematical Programming Computation}, vol.~11, no.~1, pp. 1--36, 2019.

\bibitem{kingma2014adam}
D.~P. Kingma and J.~Ba, ``Adam: A method for stochastic optimization,''
  \emph{arXiv preprint arXiv:1412.6980}, 2014.

\bibitem{zhang2018noisy}
G.~Zhang, S.~Sun, D.~Duvenaud, and R.~Grosse, ``Noisy natural gradient as
  variational inference,'' in \emph{International conference on machine
  learning}.\hskip 1em plus 0.5em minus 0.4em\relax PMLR, 2018, pp. 5852--5861.

\end{thebibliography}
\begin{IEEEbiography}[{\includegraphics[width=1in,height=1.25in,clip,keepaspectratio]{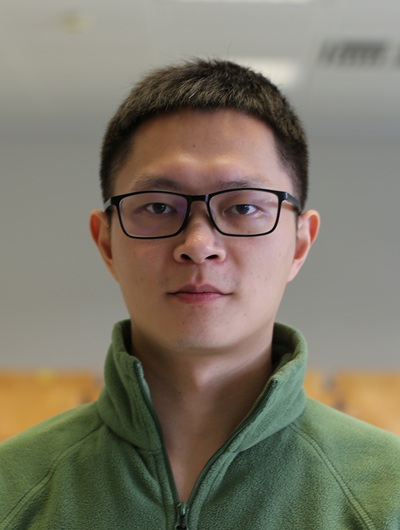}}]{Yi Yang}  received both the B.Eng. and the M.Sc. degrees in control science and engineering from Beijing Institute of Technology, China, in 2021 and 2024, respectively. He is currently working towards the PhD degree with Institute of Automatic Control, Leibniz University Hannover, Germany. 
	
	His research interests include moving horizon estimation and active learning of nonlinear systems.
\end{IEEEbiography}
\begin{IEEEbiography}[{\includegraphics[width=1in,height=1.25in,clip,keepaspectratio]{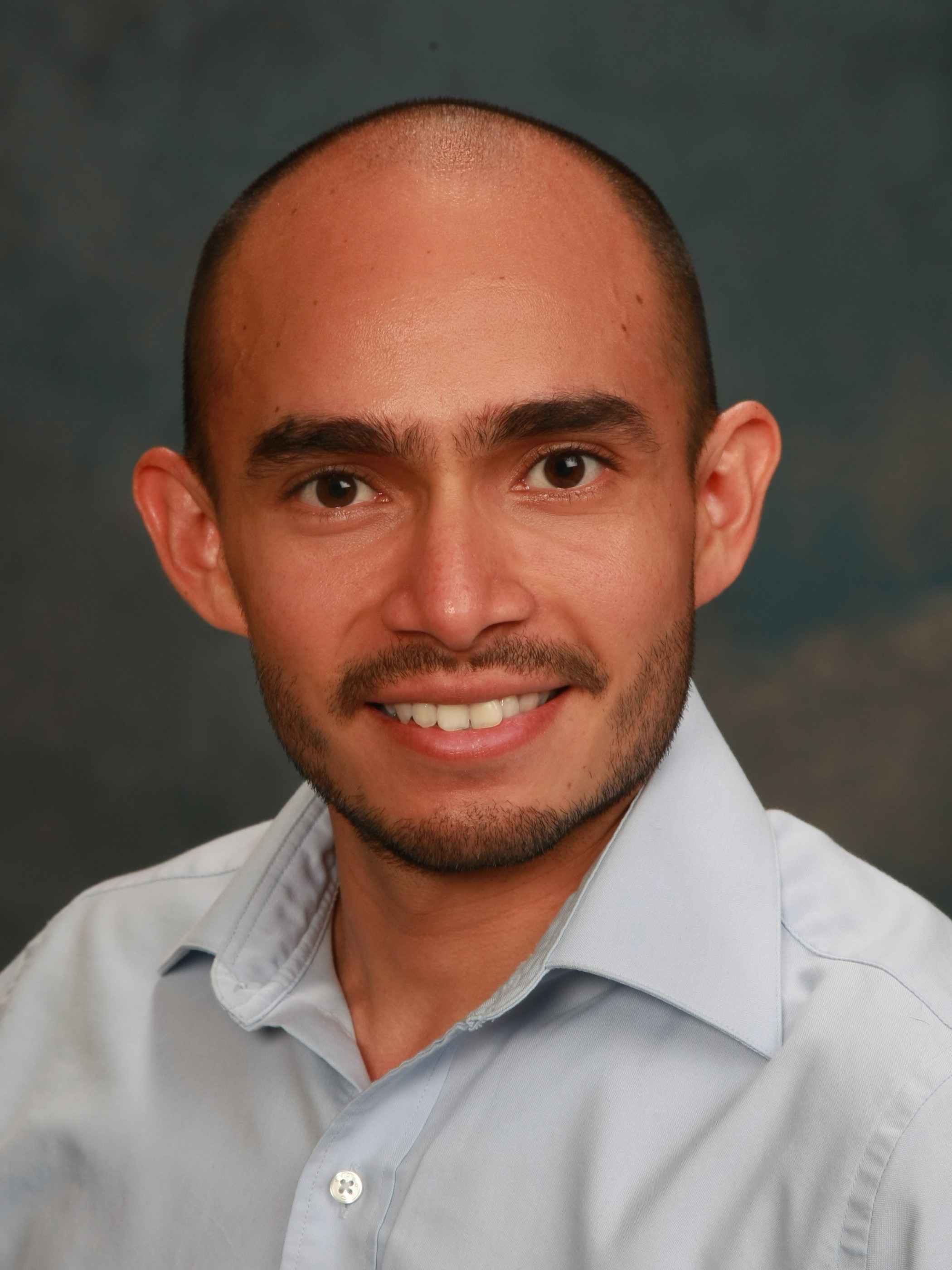}}]{Victor G. Lopez}  received a B.Sc. degree in communications and electronics engineering from the Universidad Autonoma de Campeche, Mexico, in 2010, the M.Sc. degree in electrical engineering from the Research and Advanced Studies Center (Cinvestav), Mexico, in 2013, and the PhD degree in electrical engineering from the University of Texas at Arlington, USA, in 2019. He is currently a postdoctoral researcher at the Institute of Automatic Control, Leibniz University Hannover, Germany. His research interest include data- and learning-based control, game theory, and distributed control. 
\end{IEEEbiography}
\begin{IEEEbiography}[{\includegraphics[width=1in,height=1.25in,clip,keepaspectratio]{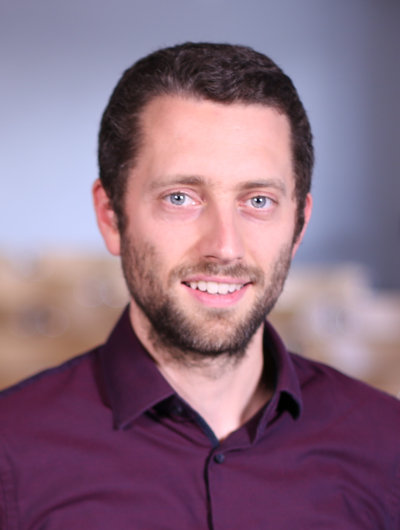}}]{Matthias A. M\"{u}ller}  received a Diploma degree in engineering cybernetics from the University of Stuttgart, Germany, an M.Sc. in electrical and computer engineering from the University of Illinois at Urbana-Champaign (both in 2009), and a Ph.D. in mechanical engineering from the University of Stuttgart in 2014. Since 2019, he is Director of the Institute of Automatic Control and Full Professor at the Leibniz University Hannover, Germany.
	
His research interests include nonlinear control and estimation, model predictive control, and data- and learning-based control, with application in different fields including biomedical engineering and robotics. He has received various distinctions for his work, including the European Systems \& Control PhD Thesis Award, an ERC Starting Grant 	from the European Research Council, the IEEE CSS George S. Axelby Outstanding Paper Award, the Brockett-Willems Outstanding Paper Award, and the Journal of Process Control Paper Award. He serves/d as an associate editor for Automatica and as an editor of the International Journal of Robust and Nonlinear Control. 
\end{IEEEbiography}

\end{document}